\documentclass[letterpaper,twocolumn,10pt]{article}
\usepackage{usenix2019_v3}

\usepackage{tikz}
\usetikzlibrary{fit,calc}
\usepackage{amsmath}

\usepackage{filecontents}

\usepackage{color,soul}

\usepackage{amsmath}
\usepackage{amssymb}
\usepackage{amsthm}
\usepackage{url}
\usepackage{multirow}
\usepackage{graphicx}
\usepackage{makecell}
\usepackage{subfig}
\usepackage{float}
\usepackage{algpseudocode}
\usepackage{algorithm}
\usepackage[english]{babel}
\usepackage{enumitem}
\setlist{nosep}
\usepackage{microtype}

\newtheorem{theorem}{Theorem}

\usepackage{xcolor}

\newcommand{\cF}{\mathcal{F}}

\newcommand{\cshare}[1]{\langle {#1}\rangle_C}

\newcommand{\sshare}[1]{\langle {#1}\rangle_S}

\newcommand{\Zbb}{\mathbb{Z}}
\newcommand{\Rbb}{\mathbb{R}}
\newcommand{\kcl}{k_{\mathrm{c}}}
\newcommand{\knn}{k_{\mathrm{nn}}}
\newcommand{\bcoord}{b_{\mathrm{c}}}
\newcommand{\bdist}{b_{\mathrm{d}}}
\newcommand{\bpid}{b_{\mathrm{pid}}}
\newcommand{\bcid}{b_{\mathrm{cid}}}
\newcommand{\rcenters}{r_{\mathrm{c}}}
\newcommand{\rpoints}{r_{\mathrm{p}}}
\newcommand{\lstash}{l_{\mathrm{s}}}
\newcommand{\uall}{u_{\mathrm{all}}}
\newcommand{\MIN}{\mathrm{MIN}}

\newcommand{\Fatopk}{\cF_{\mathrm{aTOPk}}}
\newcommand{\Ftopk}{\cF_{\mathrm{TOPk}}}

\newtheorem{definition}{Definition}

\newcommand{\bZ}{\mathbb{Z}}

\newcommand{\Fcl}{\cF_{\mathrm{ANN_{cl}}}}
\newcommand{\Fls}{\cF_{\mathrm{ANN_{ls}}}}
\newcommand{\Pls}{\Pi_{\mathrm{ANN_{ls}}}}
\newcommand{\Pcl}{\Pi_{\mathrm{ANN_{cl}}}}
\newcommand{\bR}{\mathbb{R}}

\newcommand{\ID}{\mathrm{ID}}

\newcommand{\pt}[2]{{\bf #1}_{#2}}

\newcommand*{\tikzmkone}[1]{\tikz[remember picture,overlay,] \node (#1) {};\ignorespaces}
\newcommand{\boxitone}[1]{\tikz[remember picture,overlay]{\node[yshift=-7pt,fill=#1,opacity=.25,fit={($(A)-(.28\linewidth,-0.1\baselineskip)$)($(B)+(.57\linewidth,0.6\baselineskip)$)}] {};}\ignorespaces}

\newcommand*{\tikzmkonetwo}[1]{\tikz[remember picture,overlay,] \node (#1) {};\ignorespaces}
\newcommand{\boxitonetwo}[1]{\tikz[remember picture,overlay]{\node[yshift=-7pt,fill=#1,opacity=.25,fit={($(A)-(.264\linewidth,-0.1\baselineskip)$)($(B)+(.625\linewidth,0.6\baselineskip)$)}] {};}\ignorespaces}

\newcommand*{\tikzmktwo}[1]{\tikz[remember picture,overlay,] \node (#1) {};\ignorespaces}
\newcommand{\boxittwo}[1]{\tikz[remember picture,overlay]{\node[yshift=-7pt,fill=#1,opacity=.25,fit={($(A)-(.16\linewidth,-0.1\baselineskip)$)($(B)+(.014\linewidth,0.5\baselineskip)$)}] {};}\ignorespaces}

\newcommand*{\tikzmkthree}[1]{\tikz[remember picture,overlay,] \node (#1) {};\ignorespaces}
\newcommand{\boxitthree}[1]{\tikz[remember picture,overlay]{\node[yshift=-7pt,fill=#1,opacity=.25,fit={($(A)-(.3\linewidth,-0.1\baselineskip)$)($(B)+(.51\linewidth,0.42\baselineskip)$)}] {};}\ignorespaces}

\newcommand*{\tikzmkthreeagain}[1]{\tikz[remember picture,overlay,] \node (#1) {};\ignorespaces}
\newcommand{\boxitthreeagain}[1]{\tikz[remember picture,overlay]{\node[yshift=-7pt,fill=#1,opacity=.25,fit={($(A)-(.287\linewidth,0.1\baselineskip)$)($(B)+(.59\linewidth,0.5\baselineskip)$)}] {};}\ignorespaces}

\newcommand*{\tikzmkfour}[1]{\tikz[remember picture,overlay,] \node (#1) {};\ignorespaces}
\newcommand{\boxitfour}[1]{\tikz[remember picture,overlay]{\node[yshift=-7pt,fill=#1,opacity=.25,fit={($(A)-(.138\linewidth,-0.1\baselineskip)$)($(B)+(-.011\linewidth,0.33\baselineskip)$)}] {};}\ignorespaces}

\newcommand*{\tikzmkfive}[1]{\tikz[remember picture,overlay,] \node (#1) {};\ignorespaces}
\newcommand{\boxitfive}[1]{\tikz[remember picture,overlay]{\node[yshift=-7pt,fill=#1,opacity=.25,fit={($(A)-(.24\linewidth,-0.1\baselineskip)$)($(B)+(.587\linewidth,0.5\baselineskip)$)}] {};}\ignorespaces}

\newcommand*{\tikzmkfiveagain}[1]{\tikz[remember picture,overlay,] \node (#1) {};\ignorespaces}
\newcommand{\boxitfiveagain}[1]{\tikz[remember picture,overlay]{\node[yshift=-7pt,fill=#1,opacity=.25,fit={($(A)-(.287\linewidth,0.03\baselineskip)$)($(B)+(.656\linewidth,0.5\baselineskip)$)}] {};}\ignorespaces}

\newcommand*{\tikzmksix}[1]{\tikz[remember picture,overlay,] \node (#1) {};\ignorespaces}
\newcommand{\boxitsix}[1]{\tikz[remember picture,overlay]{\node[yshift=-7pt,fill=#1,opacity=.25,fit={($(A)-(.15\linewidth,-0.15\baselineskip)$)($(B)+(1.0\linewidth,1.2\baselineskip)$)}] {};}\ignorespaces}

\newcommand*{\tikzmkseven}[1]{\tikz[remember picture,overlay,] \node (#1) {};\ignorespaces}
\newcommand{\boxitseven}[1]{\tikz[remember picture,overlay]{\node[yshift=-7pt,fill=#1,opacity=.25,fit={($(A)-(.145\linewidth,0)$)($(B)+(.53\linewidth,0)$)}] {};}\ignorespaces}

\colorlet{he_color}{blue!40}
\colorlet{gc_color}{yellow!60}
\colorlet{oram_color}{red!60}

\begin{document}

\date{}

\title{\Large \bf \textsf{SANNS:} Scaling Up Secure Approximate\\ $k$-Nearest Neighbors Search}

\author{
{\rm Hao Chen}\\
Microsoft Research
\and
{\rm Ilaria Chillotti}\\
imec-COSIC KU Leuven \& Zama
\and
{\rm Yihe Dong}\\
Microsoft
\and
{\rm Oxana Poburinnaya}\\
Simons Institute
\and
{\rm Ilya Razenshteyn}\\
Microsoft Research
\and
{\rm M.\ Sadegh Riazi}\\
UC San Diego
}

\maketitle

\begin{abstract}
The $k$-Nearest Neighbor Search ($k$-NNS) is the backbone of several cloud-based services such as recommender systems, face recognition, and database search on text and images. 
In these services, the client sends the query to the cloud server and receives the response in which case the query and response are revealed to the service provider. 
Such data disclosures are unacceptable in several scenarios due to the sensitivity of data and/or privacy laws.

In this paper, we introduce \textsf{SANNS}, a system for secure $k$-NNS that keeps client's query and the search result confidential. 
\textsf{SANNS} comprises two protocols: an optimized linear scan and a protocol based on a novel \emph{sublinear time} clustering-based algorithm. 
We prove the security of both protocols in the standard semi-honest model. 
The protocols are built upon several state-of-the-art cryptographic primitives such as lattice-based additively homomorphic encryption, distributed oblivious RAM, and garbled circuits.
We provide several contributions to each of these primitives which are applicable to other secure computation tasks. 
Both of our protocols rely on a new circuit for the approximate top-$k$ selection from $n$~numbers that is built from $O(n + k^2)$ comparators.

We have implemented our proposed system and performed extensive experimental results on four datasets in two different computation environments, demonstrating more than \emph{$18-31\times$ faster response time} compared to optimally implemented
protocols from the prior work. 
Moreover, \textsf{SANNS} is the first work that scales to the database of 10 million entries, pushing the limit by more than \emph{two orders of magnitude}.

\end{abstract}

\section{Introduction}

The \emph{$k$-Nearest Neighbor Search problem} ($k$-NNS) is defined as follows. For a given $n$-point dataset $X \subset \Rbb^d$, and a query point $\pt{q}{} \in \Rbb^d$, find (IDs of) $k$ data points closest (with respect to the Euclidean distance) to the query.
The $k$-NNS has many applications in modern data analysis:
one typically starts with a dataset (images, text, etc.) and, using domain expertise together with machine learning, produces its \emph{feature vector
representation}. Then, \emph{similarity search} queries (``find $k$~objects most similar to the query'') directly translate to $k$-NNS queries in the feature space. 
Even though some applications of $k$-NNS benefit from non-Euclidean distances~\cite{asharov2018privacy}, the overwhelming majority of
applications (see~\cite{aumuller2017ann} and the references therein) utilize Euclidean distance or cosine similarity, which can be modeled as Euclidean distance on a unit sphere.

When it comes to applications dealing with sensitive information, such as medical, biological or financial data, the privacy of both the dataset and the queries needs to be ensured. Therefore, the ``trivial solution'' where the server sends the entire dataset to the client or the client sends the plaintext query to the server would not work, since we would like to protect the input from both sides. Such settings include: face recognition~\cite{ppface,sadeghi2009efficient}, biometric identification~\cite{evans2011efficient,fingercode,demmler2015aby},
patient data search in a hospital~\cite{shaul2018scalable,asharov2018privacy} and many others.
One can pose the \emph{Secure $k$-NNS} problem, which has the same functionality as the $k$-NNS problem, and at the same time preserves the privacy of the input: the server---who holds the dataset---should learn nothing about the query or the result, while the client---who has the query---should not learn anything about the dataset
besides the $k$-NNS result. 

Secure $k$-NNS is a heavily studied problem in a variety of settings (see Section~\ref{sec_related_works} for the related work). In this paper, we consider one of the most conservative security requirements
of \emph{secure two-party computation}~\cite{evans2018pragmatic}, where the protocol is not allowed to reveal anything beyond the \emph{output} of the respective \emph{plaintext} $k$-NNS algorithm.
Note that we do not rely on a trusted third party (which is hardly practical) or trusted hardware such as Intel SGX\footnote{While the trust model of cryptographic solutions is based on computational hardness assumptions, Trusted Execution Environments (TEE)-based methodologies, such as Intel SGX, require remote attestation before the computation can begin. As a result, TEE-based solutions need to trust the hardware vendor as well as TEE implementation.} (which is known to have major security issues: see, e.g.,~\cite{van2018foreshadow}).

In this paper, we describe \textsf{SANNS}: a system for fast processing of secure $k$-NNS queries that works in the two-party secure computation setting.
The two main contributions underlying \textsf{SANNS} are the following.
First, we provide an improved secure protocol for the top-$k$ selection.
Second,
we design a new $k$-NNS algorithm tailored to secure computation,
which is implemented
using a combination of Homomorphic Encryption (HE), Garbled Circuits (GC) and Distributed Oblivious RAM (DORAM)
as well as the above top-$k$ protocol. Extensive experiments on real-world image and text data show that \textsf{SANNS} achieves a speed-up of up to $31\times$ compared to (carefully implemented and heavily optimized) algorithms from the prior work.

\vspace{0.3em}
\noindent {\bf Trust model} We prove simulation-based security of \textsf{SANNS} in the semi-honest model, where both parties follow the protocol specification while
trying to infer information about the input of the other party from the received messages.
This is an appropriate model for parties that in general trust each other (e.g., two companies or hospitals) but need to run
a secure protocol due to legal restrictions. Most of the instances of secure multi-party computation
deployed in the real world operate in the semi-honest model: 
computing gender pay gap~\cite{bestavros2017user}, sugar beets auctions~\cite{bogetoft2009secure}, and others. Our protocol yields a substantial improvement over prior works under the same trust model.
Besides, any semi-honest protocol can be reinforced to be maliciously secure (when parties
are allowed to tamper actively with the sent messages), though it incurs a significant performance overhead~\cite{goldreich1987play}.

\subsection{Specific Contributions}

Underlying \textsf{SANNS} are two new algorithms for the $k$-NNS problem.
The first one is based on \emph{linear scan,} where we compute distances to all the points, and then select the $k$ closest ones.
The improvement comes from the new top-$k$ selection protocol.
The second algorithm has \emph{sublinear} time avoiding computing all the distances.
At a high level, it proceeds by clustering the dataset using the $k$-means algorithm~\cite{lloyd1982least}, then, given a query point, we compute several closest clusters,
and then compute $k$ closest points within these clusters.
The resulting points are \emph{approximately} closest; it is known that approximation is necessary for \emph{any} sublinear-time $k$-NNS algorithm~\cite{razenshteyn2017high}\footnote{At the same time, approximation is often acceptable in practice, since feature vectors are themselves merely approximation of the ``ground truth''}.
In order to be suitable for secure computation, we introduce a new cluster \emph{rebalancing} subroutine, see below.
Let us note that among the \emph{plaintext} $k$-NNS algorithms, the clustering approach is far from being the best~\cite{aumuller2017ann}, but we find it to be particularly suitable for secure computation.

For both algorithms, we use Additive Homomorphic Encryption (AHE) for secure distance computation and garbled circuit for the top-$k$ selection.
In case of our sublinear-time algorithm, we also use DORAM to securely retrieve the clusters closest to the query.
For AHE, we use the SEAL library~\cite{sealcrypto} which implements the Brakerski/Fan-Vercauteren (BFV) scheme~\cite{FV12}. For GC we use our own implementation of Yao's protocol~\cite{Yao86} with the standard optimizations~\cite{beaver1990round,KS08,bellare2013efficient,ZRE15},
and for DORAM we implement Floram~\cite{doerner2017scaling} in the read-only mode.

Our specific contributions can be summarized as follows:
\begin{itemize}[leftmargin=0mm]
\item We propose a novel mixed-protocol solution based on AHE, GC, and DORAM that is tailored for secure $k$-NNS and achieves more than $31\times$ performance improvement compared to prior art with the same security guarantees.
\item We design and analyze an improved circuit for approximate top-$k$ selection.
The secure top-$k$ selection protocol within \textsf{SANNS} is obtained by garbling this circuit. This improvement is likely to be of independent interest for a range of other secure computation tasks. 
\item We create a clustering-based algorithm that outputs \emph{balanced} clusters, which significantly reduces the overhead of oblivious RAMs for secure random accesses. 
\item We build our system and evaluate it on various real-world datasets of text and images. We run experiments on two computation environments that represent fast and slow network connections in practice.
\item We make several optimizations to the AHE, GC, and DORAM cryptographic primitives to improve efficiency of our protocol. 
Most notably, in Floram~\cite{doerner2017scaling}, we substitute block cipher for stream cipher, yielding a speed-up by more than an order of magnitude. \end{itemize}

\subsection{Related Work}
\label{sec_related_works}

To the best of our knowledge, all prior work on the secure $k$-NNS problem in the secure two-party computation setting
is based on the \emph{linear scan}, where we first compute the distance between the query and all of $n$ database points, and then select $k$ smallest of them.
To contrast, our clustering-based algorithm is \emph{sublinear}, which leads to a substantial speed-up.
We classify prior works based on the approaches used for distance computation and for top-$k$ selection.

\vspace{0.3em}
\noindent {\bf Distance computation} \textsf{SANNS} computes distances using the BFV scheme~\cite{FV12}.
Alternative approaches used in the prior work are:
\begin{itemize}[leftmargin=0mm]
\item Paillier scheme~\cite{Pai99} used for $k$-NNS in~\cite{ppface,sadeghi2009efficient,evans2011efficient,fingercode,elmehdwi2014secure}. Unlike the BFV scheme, Paillier scheme does not support massively vectorized SIMD operations, and,
in general, is known to be much slower than the BFV scheme for vector/matrix operations such as
a batched Euclidean distance computation: see, e.g.,~\cite{gazelle}.
\item OT-based multiplication is used for $k$-NNS
in~\cite{demmler2015aby} for $k = 1$. Compared to the BFV scheme, OT-based approach requires much more
communication, $O(n + d)$ vs.\ $O(nd)$, respectively, while being slightly less compute-intensive.
In our experiments, we find that the protocol from~\cite{mohassel2017secureml} that is carefully tailored
to the matrix operations (and is, thus, significantly faster than the generic one used in~\cite{demmler2015aby}) is as fast as AHE on the fast network, but significantly slower on the slow network.
\end{itemize}

\vspace{0.3em}
\noindent {\bf Top-$k$ selection}
\textsf{SANNS} chooses $k$ smallest distances out of~$n$ by garbling a new top-$k$ circuit that we develop in this work.
The circuit is built from $O(n + k^2)$ comparators.
Alternative approaches in the prior work are:
\begin{itemize}[leftmargin=0mm]
\item The naive circuit of size $\Theta(nk)$ (c.f. Algorithm~\ref{alg: naivetopk})
was used for $k$-NNS in~\cite{asharov2018privacy,songhori2015compacting,schoppmann2018private}.
This gives asymptotically a factor of~$k$ slow-down, which is significant
even for $k = 10$ (which is a~typical setting used in practice).
\item Using homomorphic encryption (HE) for the top-$k$ selection. In the works~\cite{SFR18,shaul2018scalable}, to select $k$ smallest distances,
the BGV scheme is used, which is a variant of the BFV scheme we use for distance computations.
Neither of the two schemes are suitable for the top-$k$ selection, which is a highly non-linear operation.
A more suitable HE scheme for this task would have been TFHE~\cite{chillotti2016faster}, however, it is still known to be slower than the garbled circuit approach by at least three orders of magnitude.
\end{itemize}

We can conclude the discussion as follows: our experiments show that for $k = 10$, even the linear scan version of \textsf{SANNS}
is at up to $3.5\times$ faster than all the prior work \emph{even if we implement all the components in the prior work using the most modern tools} (for larger values of $k$, the gap would increase).
However, as we move from the linear scan to the sublinear algorithm, this yields additional speed-up up to $12\times$ 
at a cost of introducing small error in the output (on average, one out of ten reported nearest neighbors is incorrect).

All the prior work described above is in the semi-honest model except~\cite{schoppmann2018private} (which provides malicious security).
The drawback, however, is efficiency: the algorithm from~\cite{schoppmann2018private} can process one query for a dataset of size $50\,000$
in several hours. Our work yields an algorithm that can handle $10$ million data points in a matter of seconds.
All the other prior work deals with datasets of size at most $10\,000$. Thus, by designing
better algorithms and by carefully implementing and optimizing them, we scale up the datasets one can handle
efficiently by \emph{more than two orders of magnitude.}

\vspace{0.3em}
\noindent {\bf Other security models}
Some prior work considered the secure $k$-NNS problem in settings different from ``vanilla'' secure two-party computation.
Two examples are the works~\cite{riazi2016sub,wu2019privacy}. The work \cite{wu2019privacy} is under the two-server setting, which is known to give much more efficient protocols, but the security relies on the assumption that the servers do not collude.
At the same time, our techniques (e.g., better top-$k$ circuit and the balanced clustering algorithm) should yield improvements for the two-server setting as well.
In the work~\cite{riazi2016sub}, a very efficient sublinear-time protocol for secure approximate $k$-NNS is provided that provides a trade-off between privacy and the search quality. One can tune the privacy parameter to limit the information leakage based on the desired accuracy threshold. As a result, their protocol can leak more than approximate $k$-NNS results, i.e., one can estimate the similarity of two data points based on the hash values (see Section 5 of~\cite{riazi2016sub} for a~formal bound on the information leakage).

\subsection{Applications of Secure $k$-NNS} 
\textsf{SANNS} can potentially impact several real-world applications.
At a high-level, our system can provide a an efficient mechanism to retrieve similar elements to a query in any two-party computation model, e.g., database search, recommender systems, medical data analysis, etc. that provably does not leak anything beyond (approximate) answers.
For example, our system can be used to retrieve similar images within a~database given a query. 
We analyze the efficiency of our system in this scenario using the SIFT dataset which is a standard benchmark in approximate nearest-neighbor search~\cite{lowe1999object}.  
Additionally, we consider Deep1B which is a dataset of image descriptors~\cite{babenko2016efficient}. We run \textsf{SANNS} on a database as big as \emph{ten million} images, whereas the prior work deals with datasets of size at most $50\,000$. 
As another application of secure $k$-NNS consider privacy-preserving text search, which has been rigorously studied in the past~\cite{sun2013privacy,cao2013privacy,manojprivacy,pang2010privacy,gopal2012secure}. 
One group of these solutions support (multi)-keyword search~\cite{sun2013privacy,cao2013privacy,manojprivacy}: a client can receive a set of documents which include all (or subset of) keywords queried by the clients. 
In a more powerful setting, text {\it similarity} search can be performed where all documents that are semantically similar to a given document can be identified while keeping the query and the database private~\cite{pang2010privacy,gopal2012secure}. 
In this context, we evaluate \textsf{SANNS} on the Amazon reviews text database~\cite{amazon2015}.

 \section{Preliminaries}

\subsection{Secret Sharing}

In this work, we use a combination of secure computation primitives to solve the $k$-NNS problem. We connect these primitives via secret sharing, which comes in two forms: an~{\it arithmetic} secret sharing of a value $x \in \bZ_t$ is a pair $(\cshare{x}, \sshare{x})$ of random values subject to $\cshare{x} + \sshare{x} \equiv x \mod{t}$, whereas a {\it Boolean} (or XOR) secret sharing of $x \in \{0,1\}^\tau$ is a pair of random strings subject to $\cshare{x} \oplus \sshare{x}= x$. 

\label{sec: prelim}
\subsection{Distributed Oblivious RAM (DORAM)}
\label{oram_brief}

Previous solutions for secure $k$-NNS require computing distance between the query point and all points in the database, which is undesirable for large databases. In order to avoid this linear cost, we utilize a distributed version of oblivious RAM (DORAM). In this scenario, two parties hold secret shares of an array, and they can perform oblivious read and write operations, with \emph{secret-shared} indices. Typically one requires the communication cost to be sublinear in the array size. There are many known DORAM constructions~\cite{wang2014scoram,wang2015circuit,zahur2016revisiting,doerner2017scaling}, among which we choose Floram~\cite{doerner2017scaling} for efficiency reasons.  
In this work, we use Floram in {\it read-only} mode, and we further enhance its performance through careful optimizations. At a high level, we implement and use two subroutines:
\begin{itemize}[leftmargin=0mm]
    \item$\mathsf{DORAM.Init}(1^{\lambda}, DB) \to (k_A, k_B, \overline{DB})$. This step creates a masked version of the database ($\overline{DB}$) from the plaintext version ($DB$) and outputs two secret keys $k_A$ and $k_B$, one to each party. Here $\lambda$ is a security parameter.
    \item$\mathsf{DORAM.Read}(\overline{DB}, k_A, k_B, i_A, i_B) \to~(DB[i]_A,DB[i]_B)$. \\ 
    This subroutine performs the read operation where address $i$ is secret-shared between two parties as $i_A \oplus i_B = i$. Both parties acquire a XOR-share of $DB[i]$. 
\end{itemize}
In Section~\ref{ssec:oram}, we describe
these subroutines and various optimizations
in a greater detail.

\subsection{Additive Homomorphic Encryption (AHE)} 
A (private-key) additive homomorphic encryption (AHE) scheme is private-key encryption scheme with three additional algorithms $\mathsf{Add}, \mathsf{CAdd}$ and $\mathsf{CMult}$, which supports adding two ciphertexts, and addition / multiplication by constants. We require our AHE scheme to satisfy standard IND-CPA security and {\it circuit privacy}, which means that a ciphertext generated from $\mathsf{Add}$, $\mathsf{CAdd}$ and $\mathsf{CMult}$ operations should not leak more information about the
operations to the secret key owner, other than the decrypted message. This is required since in our case the server will input its secret values into $\mathsf{CAdd}$ and $\mathsf{CMult}$. We chose to use the BFV scheme~\cite{FV12}, and we achieve circuit privacy through noise flooding \cite{gazelle}.

\subsection{Garbled Circuit (GC)}

Garbled circuit (GC) is a technique first proposed by Yao in~\cite{Yao86} for achieving generic secure two-party computation for arbitrary Boolean circuits. Many improvements to GC have been proposed in literature, such as free-XOR~\cite{KS08} and half-gates ~\cite{ZRE15}. 
In addition, we use the fixed-key block cipher optimization for garbling and evaluation~\cite{bellare2013efficient}. Using Advanced Encryption Standard (AES) as the block cipher, we leverage Intel AES instructions for faster garbling procedure.

\subsection{$k$-means Clustering}
\label{kmeans_section}
One of our algorithms uses the $k$-means clustering algorithm~\cite{lloyd1982least} as a subroutine. It is a simple heuristic, which finds a clustering $X = C_1 \cup C_2 \cup \ldots \cup C_k$ into disjoint subsets $C_i \subseteq X$, and centers $\pt{c}{1}, \pt{c}{2}, \ldots, \pt{c}{k}\in \Rbb^d$, which approximately minimizes the objective function
$
\sum_{i=1}^k \sum_{{\bf x} \in C_i} \|\pt{c}{i} - {\bf x}\|^2.
$

 \section{Plaintext $k$-NNS Algorithms}
\label{sec: plaintext alg}

\vspace{0.3em}
\noindent {\bf Optimized linear scan} Our first algorithm is
a heavily optimized implementation of the linear scan:
we compute distances from the query point to \emph{all} the data points, and then (approximately) select $\knn$ data points
closest to the query. At a high level, we will implement distance computation using AHE, while top-$k$ selection is done using GC.

Computing top-$k$ na\"\i vely would require a circuit built from $O(nk)$ comparators. Instead, we propose a new algorithm for an approximate selection of top-$k$, which allows for a smaller circuit size (see section \ref{sec: approx select}) and will help us later when we implement the top-$k$ selection securely using garbled circuit.

\vspace{0.3em}
\noindent {\bf Clustering-based algorithm} The second algorithm is based on the $k$-means clustering (see Section~\ref{kmeans_section}) and, unlike our first algorithm,
has \emph{sublinear} query time. We now give a simplified version of the algorithm, and in Section~\ref{algorithm_2_real} we explain why this simplified version is inadequate and provide a full description that leads to efficient implementation.

At a high level,
we first compute $k$-means clustering of the server's dataset
with $k=\kcl$ clusters. Each cluster $1 \leq i \leq \kcl$
is associated with its \emph{center}
$\pt{c}{i} \in \Rbb^d$. During the query stage,
we find $1 \leq u \leq \kcl$
centers that are closest to the query, where $u$ is a parameter to be chosen.
Then we compute $\knn$ data points from
the corresponding $u$-many centers, and return IDs of these points as a final answer.

\subsection{Approximate Top-$k$ Selection}
\label{sec: approx select}

In both of our algorithms, we rely extensively on the following \emph{top-$k$ selection} functionality which we denote by $\MIN_n^k(x_1, x_2, \ldots, x_n)$: given a list of $n$ numbers $x_1, x_2, \ldots, x_n$, output $k \leq n$ smallest list elements in the sorted order. We can also consider the augmented functionality where each value is associated with an ID, and we output the IDs together with the values of the smallest $k$ elements. We denote this augmented functionality by $\overline{\MIN}_n^k$. In the RAM model, computing $\MIN_n^k$ is a well-studied problem, and it is by now a standard fact that it can be computed in time $O(n + k \log k)$~\cite{blum1973time}. However, to perform top-$k$ selection securely, we need to implement it as a Boolean \emph{circuit}. Suppose that all the list elements are $b$-bit integers. Then the required circuit has $b n$ inputs and $b k$ outputs. To improve efficiency, it is desirable to design a circuit for $\MIN_n^k$ with as few gates as possible.

\vspace{0.3em}
\noindent {\bf The na\"\i ve construction} 
A na\"ive circuit for $\MIN_n^k$ performs $O(nk)$ comparisons and hence consists of $O(bnk)$ gates. Algorithm~\ref{alg: naivetopk} gives such a circuit (to be precise, it computes the augmented functionality $\overline{\MIN}_n^k$, but can be easily changed to compute only ${\MIN}_n^k$). Roughly, it keeps a sorted array of the current $k$ minima. For every $x_i$, it uses a ``for'' loop to insert $x_i$ into its correct location in the array, and discards the largest item to keep it of size $k$.

\vspace{0.3em}
\noindent {\bf Sorting networks} Another approach is to employ sorting networks (e.g., AKS~\cite{ajtai19830} or the Zig-Zag sort~\cite{goodrich2014zig}) with $O(bn \log n)$ gates, which can be further improved
to $O(bn \log k)$. However, these constructions are not known to be practical.

\vspace{0.3em}
\noindent {\bf Approximate randomized selection} We are not aware of any circuit for $\MIN_n^k$ with $O(bn)$ gates unless $k$ is a constant ($O(bn)$ gates is optimal since the input has $bn$ bits). Instead, we propose a \emph{randomized} construction of a circuit with $O(bn)$ gates. 
We start with shuffling the inputs in a \emph{uniformly random order}. Namely,
instead of $x_1, x_2, \ldots, x_n$, we consider the list $x_{\pi(1)}, x_{\pi(2)}, \ldots, x_{\pi(n)}$, where $\pi$ is a uniformly random permutation of $\{1, 2, \ldots, n\}$. We require the output to be ``approximately correct'' (more on the precise definitions later) with high probability over $\pi$ for every particular list $x_1, x_2, \ldots, x_n$. 

We proceed by partitioning the input list into $l \leq n$ bins of size $n / l$ as follows:
$U_1 = \{x_{\pi(1)}, \ldots, x_{\pi(n/l)}\}$,
$U_2 = \{x_{\pi(n / l + 1)}, \ldots, x_{\pi(2 n / l)}\}$,
\ldots,
$U_l = \{x_{\pi((l - 1) n / l + 1)}, \ldots, x_{\pi(n)}\}$.
Our circuit works in two stages: first, we compute the minimum within each bin $M_i = \min_{x \in U_i} x$, then we output $\MIN_l^k(M_1, M_2, \ldots, M_l)$ as a final result using the na\"\i ve circuit for $\MIN_l^k$. The circuit size is $O(b \cdot (n + kl))$, which is
$O(bn)$ whenever $kl = O(n)$.

Intuitively, if we set the number of bins $l$ to be large enough, the above circuit should output
a high-quality answer with high probability over $\pi$. 
We state and prove two theorems formalizing this intuition in two different ways.
We defer the proofs to Appendix~\ref{appendix_topk_proofs}.

\begin{algorithm}[t]
\begin{algorithmic}
\caption{Naive Top-$k$ Computation} \label{alg: naivetopk}
\Function{NaiveTopk}{$(x_1, \ID_1), \ldots, (x_n, \ID_n)$, $k$}
\State $OPT = [\mathsf{MAXVAL}]_k$
\State $idlist = [0]_k$
\For{$i \gets 1 \ldots n$} 
	\State $x \gets x_i$, $idx \gets \ID_i$ 
	\For{$j \gets 1 \ldots k$} 
		\State $b \gets (x < OPT[j])$
		\State $(OPT[j], x)  = \Call{Mux}{OPT[j], x, b}$
		\State $(idlist[j], idx)  =\Call{Mux}{idlist[j], idx, b}$
\EndFor
\EndFor
\State \Return $(OPT, idlist)$
\EndFunction
\Function{Mux}{$a_1, a_2, b$}
\State \# Returns $(a_1, a_2)$ for $b = 0$, and $(a_2, a_1)$ for $b=1$ 
\State	\Return $(a_1 + (a_2 - a_1) \cdot b, a_2 + (a_1 - a_2) \cdot b)$
\EndFunction
\end{algorithmic}
\end{algorithm}

\begin{algorithm}[t]
\begin{algorithmic}
\caption{Approximate top-$k$ selection} \label{alg: approxtopk}
\Function{ApproxTopk}{$(x_1, \ID_1), \ldots, (x_n, \ID_n)$, $k$, $l$}
\For{$i \gets 1 \ldots l$} 
\State $(M_i, \widetilde{\ID}_i) \gets$
\State \qquad $\gets \Call{MIN}{\{(x_{(i\cdot n / l + j)}, \ID_{(i\cdot n / l + j)})\}_{j=1}^{n/l}}$ 
\EndFor
\State \Return \Call{NaiveTopk}{$(M_1, \widetilde{\ID}_1), \ldots, (M_l, \widetilde{\ID}_l)$, $k$}
\EndFunction
\end{algorithmic}
\end{algorithm}

\begin{algorithm}[t]
\begin{algorithmic}
\caption{Plaintext linear scan} \label{alg: linearscan}
\Function{LinearScanKNNS}{${\bf q}, \{{\bf p}_i \}_{i=1}^n$, $\ID$}
\State \# Uses hyperparameters $\rpoints$, $\knn$, $\lstash$ from Figure~\ref{fig:hyperparam}
\State Randomly permute the set  $\{ {\bf p}_i \}$
\For{$i \gets 1, \ldots, n$}
\tikzmkone{A}
\State $d_i \gets \|\pt{q} - \pt{p}{i}\|^2$
\tikzmkone{B}\boxitone{he_color}
\tikzmkonetwo{A}
\State $d_i \gets \lfloor \frac{d_i}{2^{\rpoints}} \rfloor$
\tikzmkonetwo{B}\boxitonetwo{gc_color}
\EndFor
\tikzmktwo{A}
\State
$(v_1, \ID_1), \ldots, (v_{\knn}, \ID_{\knn}) \gets$
\State $\,\,\Call{ApproxTopk}{d_1, \ID({\bf p}_1), \ldots, (d_n, \ID({\bf p}_n), \knn, \lstash}$
\tikzmktwo{B}\boxittwo{gc_color}
\State \Return{$\ID_1, \ldots, \ID_{\knn}$}
\EndFunction
\end{algorithmic}
\end{algorithm}

\begin{theorem}
\label{topk_expectation}
    Suppose the input list $(x_1, \ldots, x_n)$ is in uniformly random order.  There exists $\delta_0 > 0$
    and a positive function $k_0(\delta)$
    with the following property.
    For every $n$, $0 < \delta < \delta_0$,
    and $k \geq k_0(\delta)$,
    one can set the number of bins $l = k / \delta$
    such that the intersection $\mathcal{I}$ of the output of Algorithm~\ref{alg: approxtopk}
        with
    $\overline{\MIN}_n^k(x_1, x_2, \ldots, x_n)$ contains at least $(1 - \delta) k$ entries in expectation over the choice of $\pi$.
\end{theorem}

This bound yields a circuit of size $O(b \cdot (n + k^2 / \delta))$.

\begin{theorem}
\label{topk_whp}
    Suppose the input list $(x_1, \ldots, x_n)$ is in uniformly random order. 
    There exists $\delta_0 > 0$
    and a positive function $k_0(\delta)$
    with the following property.
    For every $n$, $0 < \delta < \delta_0$,
    and $k \geq k_0(\delta)$,
    one can set the number of bins $l = k^2 / \delta$
    such that
    the output of Algorithm~\ref{alg: approxtopk}
    is \emph{exactly} $\overline{\MIN}_n^k(x_1, x_2, \ldots, x_n)$ with probability at least $1 - \delta$ over the choice of $\pi$.
\end{theorem}

This yields a circuit of size $O(b \cdot (n + k^3 / \delta))$, which is worse than the previous bound, but the corresponding correctness guarantee is stronger.

\begin{figure}[t]
    \centering
    \resizebox{\columnwidth}{!}{
    \begin{tabular}{|c|c|l|}
    \hline
     & \textbf{Parameter} & \multicolumn{1}{c|}{\textbf{Description}}\\
    \Xhline{3\arrayrulewidth}

    \multirow{4}{*}{\rotatebox[origin=c]{90}{\textbf{Dataset}}}
    & $n$ & number of data points in the dataset\\
    \cline{2-3}
    & $d$ & dimensionality of the data points \\
    \cline{2-3}
    & $\knn$ & \begin{tabular}[c]{@{}l@{}} number of data points we need to return\\ as an answer \end{tabular} \\
    \Xhline{3\arrayrulewidth}

    \multirow{12}{*}{\rotatebox[origin=c]{90}{\textbf{Clustering Algorithm}}}
    & $T$ & number of \emph{groups} of clusters\\
    \cline{2-3}
    & $\kcl^i$ & \begin{tabular}[c]{@{}l@{}}total number of clusters for\\ the $i$-th group, $1 \leq i \leq T$ \end{tabular}\\
    \cline{2-3}
    & $m$ & \emph{largest} cluster size \\
    \cline{2-3}
    & $u^i$ & \begin{tabular}[c]{@{}l@{}} number of closest clusters we retrieve\\ for the $i$-th group, $1 \leq i \leq T$\end{tabular} \\
    \cline{2-3}
    & $\uall = \sum_{i=1}^T u^i$ & total number of clusters we retrieve \\
    \cline{2-3}
    & $l^i$ & \begin{tabular}[c]{@{}l@{}} is the number of bins we use to speed up \\ the selection of closest clusters for \\ the $i$-th group, $1 \leq i \leq T$ \end{tabular}\\
    \cline{2-3}
    & $\alpha$ & \begin{tabular}[c]{@{}l@{}}the allowed fraction of points in large \\ clusters during the preprocessing\end{tabular} \\
    \Xhline{3\arrayrulewidth}

    \multirow{3}{*}{\rotatebox[origin=c]{90}{\textbf{Stash}}}
    & $s$ & size of the \emph{stash}\\
    \cline{2-3}
    & $\lstash$ & \begin{tabular}[c]{@{}l@{}}number of bins we use to speed up\\ the selection of closest points for the stash \end{tabular} \\
    \Xhline{3\arrayrulewidth}

    \multirow{11}{*}{\rotatebox[origin=c]{90}{\textbf{Bitwidth}}} 
    & $\bcoord$ & \begin{tabular}[c]{@{}l@{}} number of bits necessary to encode \\ one \emph{coordinate} \end{tabular}\\
    \cline{2-3}
    & $\bdist$ & \begin{tabular}[c]{@{}l@{}}number of bits necessary to encode\\ one \emph{distance} ($\bdist = 2 \bcoord + \lceil \log_2 d\rceil$)\end{tabular} \\
    \cline{2-3}
    & $\bcid$ & \begin{tabular}[c]{@{}l@{}} number of bits necessary to encode \\ the index of a \emph{cluster} ($\bcid = \Bigl\lceil \log_2 \Bigl(\sum_{i=1}^T \kcl^i\Bigr)\Bigr\rceil$) \end{tabular}\\
    \cline{2-3}
    & $\bpid$ & \begin{tabular}[c]{@{}l@{}} number of bits for ID of a \emph{point} \end{tabular}\\
    \cline{2-3}
    & $\rcenters$ & \begin{tabular}[c]{@{}l@{}} number of bits we discard when computing \\ distances to \emph{centers of clusters},
    $0 \leq \rcenters \leq \bdist$ \end{tabular} \\
    \cline{2-3}
    & $\rpoints$ & \begin{tabular}[c]{@{}l@{}} number of bits we discard when computing \\ distances to \emph{points},
    $0 \leq \rpoints \leq \bdist$ \end{tabular}\\
    \Xhline{3\arrayrulewidth}

    \multirow{4}{*}{\rotatebox[origin=c]{90}{\textbf{AHE}}} 
    & $N$ & the ring dimension in BFV scheme \\
    \cline{2-3}
    & $q$ & ciphertext modulus in BFV scheme \\
    \cline{2-3}
    & $t = 2^{\bdist}$ & \begin{tabular}[c]{@{}l@{}} plaintext modulus in BFV scheme and \\ the modulus for secret-shared distances \end{tabular} \\
    \Xhline{3\arrayrulewidth}

    \end{tabular}}
    \caption{List of hyperparameters.}
    \label{fig:hyperparam}
\end{figure}

\subsection{Approximate Distances}
\label{approx_distances}

To speed up the top-$k$ selection further, instead of exact distances, we will be using \emph{approximate} distances. Namely,
instead of storing full $b$-bit distances, we discard $r$ low-order bits, and the overall number of gates in the selection circuit
becomes $O((b - r) \cdot (n + kl))$.
For the clustering-based algorithm, we set $r$ differently depending on whether we select closest cluster centers or closest data points,
which allows for a more fine-grained parameter tuning.

\subsection{Balanced Clustering and Stash}
\label{algorithm_2_real}

\begin{figure*}[ht]
\centering
    \includegraphics[width=0.95\textwidth]{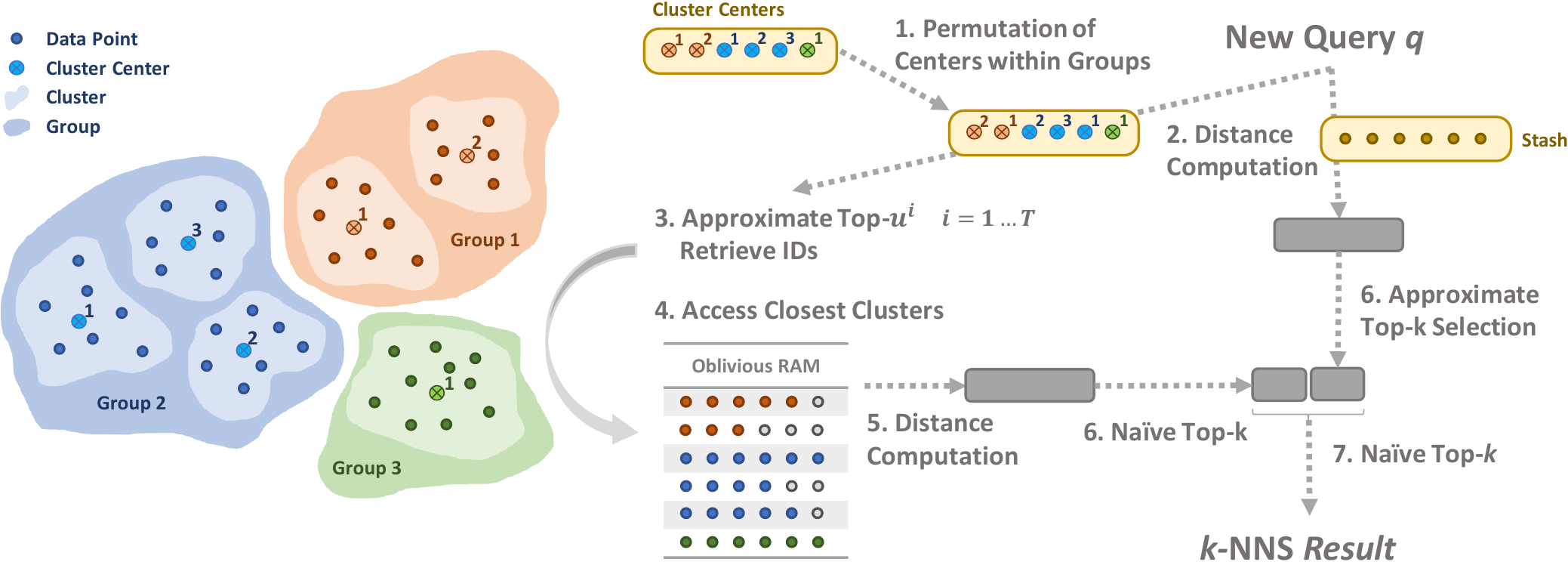}
    \caption{Visualization of \textsf{SANNS} clustering-based algorithm.}
    \label{fig:global}
\end{figure*}

To implement the clustering-based $k$-NNS algorithm securely while avoiding linear cost, we use DORAM for retrieval of clusters. In order to prevent leaking the size of each cluster, we need to set the memory block size equal to the size of the \emph{largest} cluster. This can be very inefficient, if clusters are not very balanced, i.e., the largest cluster is much larger than a \emph{typical} cluster. Unfortunately, this is exactly what
we observe in our experiments. Thus, we need a mechanism to mitigate imbalance of clusters. Below we describe one such approach, which constitutes the \emph{actual} version of the clustering-based algorithm we securely implement. With cluster balancing, our experiments achieve $3.3\times$ to $4.95\times$ reduction of maximum cluster sizes for different datasets.

We start with specifying the desired largest cluster size $1 \leq m \leq n$ and an auxiliary parameter $0 < \alpha < 1$, where $n$ denotes the total number of data points. Then,  we find the smallest $k$ (recall $k$ denotes the number of centers) such that in the clustering of the dataset $X$ found by the $k$-means algorithm at most $\alpha$-fraction of the dataset lies in clusters of size more than $m$. Then we consider all the points that belong to the said large clusters, which we denote by $X'$, setting $n' = |X'| \leq \alpha n$, and apply the same procedure recursively to $X'$. Specifically, we find the smallest $k$ such that the $k$-means clustering of $X'$ leaves at most $\alpha n'$ points in clusters of size more than $m$. We then cluster these points etc. The algorithm terminates whenever every cluster has size $\le m$.

At the end of the algorithm, we have $\widetilde{T}$ \emph{groups} of clusters
that correspond to disjoint subsets of the dataset (as a side remark, we note that one always has $\widetilde{T} \leq \log_{1/\alpha} n$). We denote the number of clusters in the $i$-th group by $\kcl^i$, the clusters themselves by $C^i_1, C^i_2, \ldots, C^i_{\kcl^i} \subseteq X$ and their centers by $c^i_1, c^i_2, \ldots, c^i_{\kcl^i} \in \Rbb^d$.
During the query stage, we find $u^i$ clusters from the $i$-th group with the centers closest to the query point, then we retrieve all the data points
from the corresponding $\sum_{i=1}^{\widetilde{T}} u^i$ clusters, and finally from these retrieved points we select $\knn$ data points that are closest to the query.

We now describe one further optimization that helps to speed up the resulting $k$-NNS algorithm even more. Namely, we collapse last several groups into a special set of points, which we call a \emph{stash}, denoted by $S \subseteq X$. In contrast to clusters from the remaining groups, to search the stash, we perform \emph{linear scan}. We denote $s = |S|$ the stash size and $T \leq \widetilde{T}$ the number of remaining groups of clusters that are not collapsed.

The motivation for introducing the stash is that the last few groups are usually pretty small, so in order for them to contribute to the overall
accuracy meaningfully, we need to retrieve most of the clusters from them. But this means many DORAM accesses which are less efficient than
the straightforward linear scan.

Note that while the simplified version of Algorithm~\ref{alg: linearscan} is well-known and very popular in practice (see, e.g.,~\cite{jegou2011product, johnson2017billion}), our modification of the algorithm in 
this section, to the best of our knowledge, is new.
It is interesting to observe that in the ``plaintext world'', clustering algorithm is far from being the best for $k$-NNS (see~\cite{aumuller2017ann} for the benchmarks), but several of its properties (namely, few non-adaptive memory accesses and that it requires computing many distances at once) make it very appealing for the secure computation.

\begin{algorithm}[h!]
\begin{algorithmic}
\caption{Plaintext clustering-based algorithm} \label{alg: clustering} \label{alg: plain clustering} 
\Function{ClusteringKNNS}{$\pt{q}{}$, $C_j^i$, $\pt{c}{j}^i$, $S$, $\ID$}
\State \# The algorithm uses hyperparameters in Figure~\ref{fig:hyperparam}
\State Randomly permute the cluster centers in each group
\State and all points in stash 
\For{$i \gets 1, \ldots, T$}
\For{$j \gets 1, \ldots, \kcl^i$}
\tikzmkthree{A}
\State $d_j^i \gets \|{\bf q} - \pt{c}{j}^i\|^2$
\tikzmkthree{B}\boxitthree{he_color}
\tikzmkthreeagain{A}
\State $d_j^i \gets \lfloor \frac{d_j^i}{2^{\rcenters}} \rfloor$
\tikzmkthreeagain{B}\boxitthreeagain{gc_color}
\EndFor
\tikzmkfour{A}
\State $(v_1, ind_1^i), \ldots, (v_{u^i}, ind_{u^i}^i) \gets$
\State \qquad $\gets \Call{ApproxTopk}{(d_1^i, 1), \ldots, (d^i_{\kcl^i}, \kcl^i), u^i, l^i}$
\tikzmkfour{B}\boxitfour{gc_color}
\EndFor
\tikzmkseven{A}
\State $C \gets \bigcup\limits_{1 \le i \le T} \bigcup\limits_{1 \le j \le u_i} C^i_{ind_j^i}$
\tikzmkseven{B}\boxitseven{oram_color}
\For{$\pt{p}{} \in C \cup S$}
\tikzmkfive{A}
\State $d_{\pt{p}{}} \gets \|\pt{q}{} - \pt{p}{} \|^2$
\tikzmkfive{B}\boxitfive{he_color}
\tikzmkfive{A}
\State $d_{\pt{p}{}} \gets \lfloor \frac{d_{\pt{p}{}}}{2^{\rpoints}} \rfloor$
\tikzmkfiveagain{B}\boxitfiveagain{gc_color}
\EndFor
\tikzmksix{A}
\State $(a_1, \widetilde{\ID}_1), \ldots, (a_{\knn}, \widetilde{\ID}_{\knn}) \gets$
\State \qquad $\gets \Call{NaiveTopk}{\{(d_{\pt{p}{}}, \ID(\pt{p}{}))\}_{\pt{p}{} \in C}, \knn}$
\State $(a_{\knn+1}, \widetilde{\ID}_{\knn+1}), \ldots, (a_{2\knn}
, \widetilde{\ID}_{2k}) \gets$
\State \qquad $\gets \Call{ApproxTopk}{\{(d_{\pt{p}{}}, \ID(\pt{p}{}))\}_{\pt{p}{} \in S}, \knn, \lstash}$
\State $(v_1, \widehat{\ID}_1), \ldots, (v_{\knn}, \widehat{\ID}_{\knn})) \gets$
\State \qquad $ \gets \Call{NaiveTopk}{(a_1, \widetilde{\ID}_1), \ldots, (a_{2\knn}, \widetilde{\ID}_{2\knn}), \knn}$
\tikzmksix{B}\boxitsix{gc_color}
\vspace{-0.08in}
\State \Return{$\widehat{\ID}_1, \ldots, \widehat{\ID}_{\knn}$}
\EndFunction
\end{algorithmic}
\end{algorithm}

\subsection{Putting It All Together}

We now give a high-level summary of our algorithms
and in the next section we provide a more detailed
description. For the linear scan, we use the approximate top-$k$ selection to return the $\knn$ IDs after computing distances between query and all points in the database.

For the clustering-based algorithm, we use approximate top-$k$ selection for retrieving $u^i$ clusters in $i$-th group for all $i \in \{1, \ldots, T\}$.
Then, we compute the closest $\knn$ points from the query to all the retrieved points using the naive top-$k$ algorithm. 
Meanwhile, we compute the approximate top-$k$ with $k = \knn$ among distances between query and the stash. Finally, we compute and output the $\knn$ closest points from the above $2\knn$ candidate points.

Note that in the clustering-based algorithm, we use exact top-$k$ selection for retrieved points and approximate selection for cluster centers and stash. The main reason is that the approximate selection requires input values to be shuffled. The corresponding permutation can be known only by the server and not by the client to ensure that there is no additional leakage when the algorithm is implemented securely. Jumping ahead to the secure protocol in the next section, the points we retrieve from the clusters will be secret-shared. Thus, performing approximate selection on retrieved points would require a secure two-party shuffling protocol, which is expensive. Therefore, we garble a \emph{na\"\i ve} circuit for exact computation of top-$k$ for the retrieved points. Figure~\ref{fig:global} visualizes \textsf{SANNS} clustering-based algorithm. 

Figure~\ref{fig:hyperparam} lists the hyperparameters used by our algorithms. See Figure~\ref{settings_both} and Figure~\ref{settings_clustering} for the values that we use for various datasets. Our plaintext algorithms are presented in Algorithm~\ref{alg: linearscan} and Algorithm~\ref{alg: plain clustering}.

 \section{Secure Protocols for $k$-NNS}
\label{sec: secure alg}

Here we describe our secure protocols for $k$-NNS. 
For the security proofs, see Appendix~\ref{sec: proof}.  The formal specifications of the protocols are given in Figure~\ref{fig: linear scan protocol} and Figure~\ref{fig: clustering protocol}.  On a high level, our secure protocols implement plaintext algorithms~\ref{alg: linearscan} and~\ref{alg: plain clustering}, which is color-coded for reader's convenience: we implemented the blue parts using AHE, yellow parts using garbled circuit, and red parts using DORAM. These primitives are connected using secret shares, and we perform share conversions (between arithmetic and Boolean) as needed.

\subsection{Ideal Functionalities for Subroutines}

Here we define three ideal functionalities $\cF_{\mathrm{TOPk}}$, $\cF_{\mathrm{aTOPk}}$, and $\cF_{\mathrm{DROM}}$ used in our protocol. We securely implement the first two using garbled circuits, and the third using Floram~\cite{doerner2017scaling}.

\begin{figure}[ht!] \footnotesize
	\framebox{\begin{minipage}{0.95\linewidth}
			Parameters: array size $m$, modulus $t$, truncation bit size $r$, output size $k$, bit-length of ID $\bpid$
			
		        Extra parameter: $\textsf{returnVal} \in \{false,true\}$ (if set to true, return secret shares of (value, ID) pairs instead of just ID. )

			\begin{itemize}
				\item On input $A_c$ and $idlist_c$ from the client, store $A_c$. 
			    \item On input $A_s$, $idlist_s$ from the server, store $A_s$ and $idlist$. 		
				\item When both inputs are received, compute $A = (A_s +A_c) \mod t = (a_1, \ldots, a_n)$ and set $a_i' = [a_i/2^r]$, $idlist = idlist_c \oplus idlist_s$. Then, let $(b,c) = \overline{\MIN}_n^k(a_1', a_2', \ldots, a_n', idlist, k)$. Sample an array $w$ of size $k$ with random entries in $\{0, 1\}^{\bpid}$, output $c \oplus w$ to the client, and $w$ to the server. If $\textsf{returnVal}$ is true, sample a random array $s$ of size $k$ in $\bZ_{2^t}$, output $b-s$ to client and $s$ to the server.
				\end{itemize}
	\end{minipage}}
	\caption{Ideal functionality $\cF_{\mathrm{TOPk}}$}
\end{figure}

\begin{figure}[ht!] \footnotesize
	\framebox{\begin{minipage}{0.95\linewidth}
			Parameters: array size $m$, modulus $t$, truncation bit size $r$,  output size $k$, bin size $l$, ID bit length $\bpid$. 
			
			Extra parameter: $\textsf{returnVal} \in \{false,true\}$ (if set to true, return (value, ID) instead of just ID. )
			\begin{itemize}
				\item On input $A_c \in \bZ_t^m$ from the client, store $A_c$. 
			    \item On input $A_s \in \bZ_t^m$ and $idlist$ from the server, store $A_s$ and $idlist$. 		
				\item When both inputs are received, compute $A = A_s +A_c \mod t = (a_1, \ldots, a_n)$. and set $a_i' = [a_i/2^r]$.
				Let $(b, c) = \mathrm{APPROXTOPK}(a_1', \ldots, a_n', idlist, k, l)$.  Sample an array $w$ of size $k$ with random entries in $\{0,1\}^{\bpid}$. Output $c \oplus w$ to the client, and $w$ to the server. If $\textsf{returnVal}$ is true, sample a random array $s$ of size $k$, output $b-s$ to client and $s$ to the server.
				\end{itemize}
	\end{minipage}}
	\caption{Ideal functionality $\cF_{\mathrm{aTOPk}}$}
\end{figure}

\begin{figure}[ht!] \footnotesize
	\framebox{\begin{minipage}{0.95\linewidth}
			Parameters: Database size $n$, bit-length of each data block $b$. 
			\begin{itemize}
				\item $\mathsf{Init}$: on input $(\mathsf{Init}, DB)$ from the server, it stores $DB$. 
			    \item $\mathsf{Read}$: on input $(\mathsf{Read}, i_c)$ and $(\mathsf{Read}, i_s)$ from both client and server, it samples a random $R \in \{0,1\}^b$. Then it outputs $DB[(i_s + i_c) \mod{n}] \oplus R$	to client and outputs $R$ to server.
			\end{itemize}
	\end{minipage}}
	\caption{Ideal functionality $\cF_{\mathrm{DROM}}$}
\end{figure}

\begin{figure}[t] \footnotesize
    \framebox{\begin{minipage}{0.95\linewidth}
            {\bf Public Parameters:} coefficient bit length $b_c$, number of items in the database $n$, dimension $d$, AHE ring dimension $N$, plain modulus $t$, ID bit length $\bpid$, bin size $l_s$. 
            
            {\bf Inputs:} client inputs query $\pt{q}{} \in \bR^d$; server inputs $n$ points and a list $idlist$ of $n$ IDs. 
            \begin{enumerate}
                \item Client calls $\mathsf{AHE.Keygen}$ to get $sk$; server randomly permutes its points. They both quantize their points into $\pt{q}{}', \pt{p}{i}' \in \bZ_{2^{b_c}}^d$.      
                \item  Client sends $c_i = \mathsf{AHE.Enc}(sk, \pt{q}{}'[i])$ for $1 \leq i \leq d$ to the server.
                \item Server sets $p_{ik} =   \pt{p}{kN+1}'[i] + \pt{p}{kN+2}'[i] x + \cdots + \pt{p}{(k+1)N}'[i] x^{N-1}$,  samples random vector ${\bf r} \in \bZ_t^n$  and computes for $1 \leq k \leq \lceil n/N \rceil$ 
                $$f_k = \sum_{i=1}^d  \mathsf{AHE.CMult}(c_i,  {\bf p}_{ik}) + {\bf r}[kN: (k+1)N]$$ . 
                \item  Server sends $f_k$ to client who decrypts them to ${\bf s} \in \bZ_t^n$. 
                \item Client sends $-2{\bf s} + ||{\bf q}'||^2 \cdot (1,1,\ldots,1)$ to $\Fatopk$, server sends $idlist$ and $(-2r_i + ||{\bf p}_i'||^2)_i$ to $\Fatopk$, with parameters $(k, l_s, false)$. They output $[{\bf id}]_c$, $[{\bf id}]_s \in \{0,1\}^{\bpid}$. Server sends $[{\bf id}]_s$ to client, who outputs ${\bf id}= [{\bf id}]_c \oplus [{\bf id}]_s$. 
                \end{enumerate}
    \end{minipage}}
    \caption{\textsf{SANNS} linear-scan protocol $\Pls$.}
    \label{fig: linear scan protocol}
\end{figure}

\begin{figure}[t] \footnotesize
    \framebox{\begin{minipage}{0.95\linewidth}
            {\bf Public Parameters}: coefficient bit length $b_c$, number of items in the database $n$, dimension $d$, AHE ring dimension $N$, plain modulus $t$.  
            
            {\bf Clustering hyperparameters}: $T$, $\kcl^i$, $m$, $u^i$, $s$, $l^i$, $\lstash$, $\bcoord$, $\rcenters$ and $\rpoints$. 
            
            {\bf Inputs}: client inputs query $\pt{q}{} \in \bR^d$; server inputs $T$ groups of clusters with each cluster of size up to $m$, and a stash $S$; server also inputs a list of $n$ IDs $idlist$, and all cluster centers ${\bf c}_j^i$. 

\begin{enumerate}
\item Client calls $\mathsf{AHE.Keygen}$ to get $sk$.
\item Client and server quantize their points and the cluster centers. 
\item 
Server sends all clusters with one block per cluster, and each point accompanied by its ID and squared norm, to $\cF_{\mathrm{DROM}}.\mathsf{Init}$, padding with dummy points if necessary to reach size $m$ for each block.
\item  The server performs two independent random shuffles on the cluster centers and stash points.
\item For each $i \in \{1, \ldots, T\}$, 
\begin{itemize}
    \item The client and server use line 3-5 in Figure~\ref{fig: linear scan protocol} to compute secret shares of the vector $(||{\bf q} -  \pt{c}{j}^i||_2^2)_j$. 
    \item Client and server send their shares to $\Fatopk$ with $k = u^i$, $l = l^i$ and \textsf{returnVal} = false, when server inputs the default $idlist = \{0,1, \ldots, k_c^i-1\}$. They obtain secret shares of indices $j_1^i, \ldots, j_{u_i}^i$. 
\end{itemize}

\item Client and server input  secret shares of all cluster indices $\{ (i, j_c^i): i \in [1,T], c \in [1, u^i] \}$ obtained in  step~5 into $\cF_{\mathrm{DROM}}.\mathsf{Read}$, to retrieve Boolean secret shares of tuples $(\pt{p}{}, \ID(\pt{p}{}), ||\pt{p}{}||^2)$ of all points in the corresponding clusters. They convert $\pt{p}{}$ and $||\pt{p}{}||^2$ to arithmetic secret shares using e.g. the B2A algorithm in \cite{demmler2015aby}. 

\item Client and server use line 3-6 in Figure~\ref{fig: linear scan protocol} to get secret shares of a distance vector for all points determined in step~6. Then, they input their shares of points and IDs to $\Ftopk$ with \textsf{returnVal} = true, and output secret shares of a list of tuples $(d_i^{Cluster}, \ID_i^{Cluster})_{i=1}^k$. 

\item For the stash $S$, client and server use line 3-6 in Figure~\ref{fig: linear scan protocol} to obtain the secret shared distance vector. Then, they input their shares (while server also inputs IDs of stash points and client input a zero array for its ID shares) to $\Fatopk$ with parameters $(k, l_s, true)$, and output shares of $(d_i^{Stash},  \ID_i^{Stash})_{i=1}^k$.

\item Each party inputs the union of shares of (point, ID) pairs obtained from steps 7-8 to $\Ftopk$ with \textsf{returnVal}=false, and outputs secret shares of $k$ IDs. Server sends its secret shares of IDs to the client, who outputs the final list of IDs.
\end{enumerate}
    \end{minipage}}
    \caption{\textsf{SANNS} clustering-based protocol $\Pcl$.} 
    \label{fig: clustering protocol}
\end{figure}

\subsection{Distance Computation via AHE}

We use the BFV scheme \cite{FV12} to compute distances. Compared to~\cite{gazelle}, which uses BFV for matrix-vector multiplications, our approach avoids expensive ciphertext \emph{rotations}. Also, we used the coefficient encoding and a plaintext space modulo a power of two instead of a prime. This allows us to later avoid a costly addition modulo $p$ inside a garbled circuit.

More precisely, SIMD for BFV requires plaintext modulus to be prime $p \equiv 1 \mod{2N}$. However, it turns out our distance computation protocol only requires multiplication between \emph{scalars} and vectors. Therefore we can drop the requirement and perform computation modulo powers of two without losing efficiency. Recall that plaintext space of the BFV scheme is $R_t := \Zbb_t[x]/(x^{N} + 1)$. The client encodes each coordinate in to a constant polynomial $f_i = {\bf q}[i]$. Assume the server points are ${\bf p}_1, \ldots, {\bf p}_N$ for simplicity. It encodes these points into $d$ plaintexts, each encoding one coordinate of all points, resulting in $g_i = \sum_j {\bf p}_{j+1}[i] x^j$. Note that $\sum_{i=1}^d f_i g_i = \sum_{j=1}^{N} \langle {\bf q}, {\bf p}_j \rangle x^{j-1}$. The client sends encryption of $f_i$. Then the server computes an encryption  $h(x) = \sum_i f_i g_i$,  masks $h(x)$ with a random polynomial and sends back to the client, so they hold secret shares of 
$\langle {\bf q}, {\bf p}_j \rangle$ modulo $t$. Then, secret shares of Euclidean distances modulo $t$ can be reconstructed via local operations. 

Note that we need to slightly modify the above routine when computing distances of points retrieved from DORAM. Since the server does not know these points in the clear, we let client and server secret share the points and their squared Euclidean norms. 

\subsection{Point Retrievals via DORAM}\label{ssec:oram}
We briefly explain the functionality of Floram and refer the reader to the original paper~\cite{doerner2017scaling} for details.

In Floram, both parties hold {\it identical} copies of the masked database. Let the plaintext database be $DB$, block at address $i$ be $DB[i]$, and the masked database be $\overline{DB}$. We set:
$$
\overline{DB}[i] = DB[i] \oplus PRF_{k_A}(i) \oplus
PRF_{k_B}(i),
$$
where $PRF$ is a pseudo-random function, $k_A$
is a secret key owned by A and $k_B$ is similarly owned by B.
At a high level, Floram's retrieval functionality consists of the two main parts: token generation using Functional Secret Sharing (FSS)~\cite{gilboa2014distributed} and data unmasking from the PRFs.
In Floram, FSS is used to securely generate two bit vectors (one for each party) $u^A$ and $u^B$ such that
individually they look random, yet $u^A_j \oplus
u^B_j = 1$ iff $j = i$, where $i$ is the
address we are retrieving.
Then, party A computes $\bigoplus_j u^A_j \cdot \overline{DB}[i]$
and, likewise, party B computes $\bigoplus_j u^B_j \cdot \overline{DB}[i]$.
The XOR of these two values is simply
$\overline{DB}[i]$.
To recover the desired value $DB[i]$, the parties use a garbled circuit to
compute the PRFs and XOR to remove
the masks.\footnote{The retrieved block can be either returned to one party, or secret-shared between the parties within the same garbled circuit}

We implemented Floram with a few optimizations described below.

\vspace{0.3em}
\noindent {\bf Precomputing OT}
To run FSS, the parties have to execute the GC protocol $\log_2 n$ times iteratively which in turn requires $\log_2 n$ set of Oblivious Transfers (OTs). Performing consecutive OTs can significantly slow down the FSS evaluation.
We use Beaver OT precomputation protocol~\cite{beaver1995precomputing} which allows to perform all necessary OTs on random values in the beginning of FSS evaluation with a very small additional communication for each GC invocation. 

\vspace{0.3em}
\noindent {\bf Kreyvium as PRF}
Floram implemented PRF using AES. While computing AES is fast in plaintext due to Intel AES instructions, it requires many gates to be evaluated within a garbled circuit. We propose a more efficient solution based on Kreyvium~\cite{Kreyvium} which requires significantly fewer number of AND gates (see Appendix~\ref{app:stream} for various related trade-offs). Evaluating Kreyvium during the initial database masking adds large overhead compared to AES. 
To mitigate the overhead, we pack multiple (512 in our case) invocations of Kreyvium and evaluate them simultaneously by using AVX-512 instructions provided by Intel CPUs. 

\vspace{0.3em}
\noindent {\bf Multi-address access}
In Floram, accessing the database at $k$ different locations requires $k \log_2 n$ number of interactions. 
In our case, these memory accesses are non-adaptive, hence we can fuse these accesses and reduce the number of rounds to $\log_2 n$ which has significant effect in practice.

\subsection{Top-$k$ Selection via Garbled Circuit}

We implement secure top-$k$ selection using garbled circuit while we made some further optimizations to improve the performance. First, we truncate distances by simply discarding some lower order bits, which allows us to reduce the circuit size significantly (see Section~\ref{approx_distances}). The second optimization comes from the implementation side. Recall that existing MPC frameworks such as ABY~\cite{demmler2015aby} require storing the entire circuit explicitly with accompanying bloated data structures. However, our top-$k$ circuit is highly structured, which allows us to work with it looking at one small part at a time. This means that the memory consumption of the garbling and the evaluation algorithms can be essentially independent of $n$, which makes them much more cache-efficient. To accomplish this, we developed our own garbled circuit implementation with most of the standard optimizations~\cite{beaver1990round,KS08,bellare2013efficient,ZRE15}\footnote{For oblivious transfer, we use libOTe~\cite{libOTe}}, which allows us to save more than an order of magnitude in both time and memory usage compared to ABY.
 \section{Implementation and Performance Results}
\label{sec: performance}

\subsection{Environment}

We perform the evaluation on two Azure \texttt{F72s\_v2} instances (with $72$ \emph{virtual} cores
equivalent to that of Intel Xeon Platinum 8168
and $144$ GB of RAM each).
We have two sets of experiments: for \emph{fast} and \emph{slow} networks.
For the former we use two instances from the ``West US 2'' availability zone (latency 0.5 ms, throughput from 500 MB/s to 7 GB/s depending on the number of simultaneous network connections),
while for the latter we run on instances hosted in ``West US 2'' and ''East US'' (latency 34 ms, throughput from 40 MB/s to 2.2 GB/s).
We use g++ 7.3.0, Ubuntu 18.04, SEAL 2.3.1~\cite{sealcrypto} and libOTe~\cite{libOTe} for the OT phase (in the single-thread mode due to unstable
behavior when run in several threads). For networking, we use ZeroMQ. We implement balanced clustering as described in Section~\ref{algorithm_2_real} using PyTorch and run it on four NVIDIA Tesla V100 GPUs. It is done once per dataset and takes several hours (with the bottleneck being the vanilla $k$-means clustering described in Section~\ref{kmeans_section}).

\subsection{Datasets}

We evaluate \textsf{SANNS} algorithms as well as baselines on four datasets: SIFT ($n=1\,000\,000$, $d=128$) is a standard dataset of image descriptors~\cite{lowe1999object} that can be used to compute similarity between images; Deep1B ($n=1\,000\,000\,000$, $d = 96$) is also a dataset of image descriptors~\cite{babenko2016efficient}, which is built from feature vectors obtained by passing images through a deep neural network
(for more details see the original paper~\cite{babenko2016efficient}), Amazon ($n = 2^{20}$, $d = 50$) is dataset of reviews~\cite{amazon2015}, where feature vectors are obtained using word embeddings. We conduct the evaluation on two subsets of Deep1B
that consist of the first $1\,000\,000$ and $10\,000\,000$ images, which we label Deep1B-1M
and Deep1B-10M, respectively. For Amazon, we take $2^{20}$ Amazon reviews of the CDs and Vinyls category, and create a vector embedding for each review by processing GloVe word embeddings \cite{glove} as in \cite{embeddings2015}. 
SIFT comes with $10\,000$ sample queries which are used for evaluation; for Deep1B-1M, Deep1B-10M
and Amazon, a sample of $10\,000$ data points from the dataset are used as queries. For all the datasets we use Euclidean distance
to measure similarity between points. Note that the Deep1B-1M and Deep1B-10M datasets are normalized to lie on the unit sphere.

Note that all four datasets have been extensively used in nearest neighbors benchmarks and information retrieval tasks. In particular, SIFT is a part of ANN Benchmarks~\cite{aumuller2017ann},
where a large array of NNS algorithms has been thoroughly evaluated. Deep1B has been
used for evaluation of NNS algorithms in, e.g., \cite{babenko2016efficient, johnson2017billion, malkov2018efficient}.
Various subsets of the Amazon dataset have been used to evaluate the accuracy and the efficiency of $k$-NN classifiers in, e.g.,
\cite{kusner2015word,dong2019scalable}.

\subsection{Parameters}

\vspace{0.3em}
\noindent {\bf Accuracy}
In our experiments, we require the algorithms to return $\knn = 10$ nearest neighbors and measure accuracy as the average portion of correctly returned points over the set of queries (``$10$-NN accuracy''). Our algorithms achieve $10$-NN accuracy at least $0.9$ ($9$ out of $10$ points are correct on average), which is a level of accuracy considered to be acceptable in practice (see, e.g.,~\cite{kulis2009kernelized,li2014two}).

\vspace{0.3em}
\noindent {\bf Quantization of coordinates}
For SIFT, coordinates of points and queries are already small integers between $0$ and $255$, so
we leave them as is. For Deep1B, the coordinates are real numbers, and we quantize them to $8$-bit integers uniformly between the minimum and the maximum values of all the coordinates. For Amazon we do the same but with $9$ bits.
For these datasets, quantization barely affects the $10$-NN accuracy compared to using the true floating point coordinates.

\vspace{0.3em}
\noindent {\bf Cluster size balancing}
As noted in Section~\ref{algorithm_2_real}, our cluster balancing algorithm achieves the crucial bound over the maximum 
cluster size needed for efficient ORAM retrieval of candidate points. In our experiments, for SIFT, Deep1B-10M, Amazon and Deep1B-1M, the balancing algorithm reduced the maximum cluster size by factors of $4.95\times$, $3.67\times$, $3.36\times$ and $3.31\times$, respectively.

\vspace{0.3em}
\noindent {\bf Parameter choices}
We initialized the BFV scheme with parameters  $N = 2^{13}$, $t = 2^{24}$ for Amazon and $t = 2^{23}$ for the other datasets, and a 180-bit modulus $q$. 
For the parameters such as standard deviation error and secret key distribution we use SEAL default values. 
These parameters allow us to use the noise flooding technique to provide 108 bits of statistical circuit privacy.\footnote{We refer the reader to \cite{gazelle} for details on the noise flooding technique} 
The LWE estimator\footnote{We used commit 3019847 from \url{https://bitbucket.org/malb/lwe-estimator}} by Albrecht et al.~\cite{estimator} suggests 141 bits of computational security.

Here is how we set the hyperparameters for our algorithms.
See Figure~\ref{fig:hyperparam} for the full list of hyperparameters, below we list the ones that affect the performance:
\begin{itemize}[leftmargin=0mm]
    \item Both algorithms depend on $n$, $d$, $\knn$, which depend on the dataset and our requirements;
    \item The linear scan depends on $\lstash$, $\bcoord$ and $\rpoints$,
    \item The clustering-based algorithm depends on $T$, $\kcl^i$, $m$, $u^i$, $s$, $l^i$, $\lstash$, $\bcoord$, $\rcenters$ and $\rpoints$,
    where $1 \leq i \leq T$.
\end{itemize}

We use the \emph{total number of AND gates} in the top-$k$ and the ORAM circuits as a proxy for both communication and running time during hyperparameter search phase (this is due to the complexity of garbling a circuit depending heavily on the number of AND gates
due to the Free-XOR optimization~\cite{KS08}). 
Moreover, for simplicity we neglect the FSS part of ORAM, since it does not affect the performance much.
Overall, we search for the hyperparameters that yield $10$-NN accuracy at least $0.9$ minimizing the total number of AND-gates.
In Figure~\ref{settings_both} and Figure~\ref{settings_clustering} of Appendix~\ref{sec:optimal_param}, we summarize the parameters we use for both algorithms
on each dataset.

\subsection{\textsf{SANNS} End-to-End Evaluation}

\vspace{0.3em}
\noindent {\bf Single-thread}
We run \textsf{SANNS} on the above mentioned four datasets using two algorithms (linear scan and clustering) over fast and slow networks
in a single-thread mode, summarizing results in Table~\ref{single_thread_table}.
We measure per-client preprocessing of Floram
separately and split the query measurements into the OT phase, distance computation, approximate top-$k$ selection and ORAM retrievals. For each of the components, we report communication and average running time for fast and slow networks. We make several observations:
\begin{itemize}[leftmargin=0mm]
  \item On all the datasets, clustering-based algorithm is much faster than linear scan: up to $12\times$ over the fast network
  and up to $8.2\times$ over the slow network.
  \item For the clustering algorithm, per-client preprocessing is very efficient. In fact, even if there is a \emph{single}
  query per client, clustering algorithm with preprocessing is faster than the linear scan.
  \item In terms of communication, distance computation part is negligible, and the bottleneck is formed by the top-$k$
  selection and ORAM (which are fairly balanced).
  \item As a result, when we move from fast to slow network, the time for distance computation stays essentially
  the same,
  while the time for top-$k$ and ORAM goes up dramatically. This makes our new circuit for approximate
  top-$k$ selection and optimizations to Floram absolutely crucial for the overall efficiency.
\end{itemize}

\vspace{0.3em}
\noindent {\bf Multi-thread}
In Table~\ref{multi_thread_table} we summarize how the performance of \textsf{SANNS}
depends on the number of threads. We only measure the query time excluding the OT phase,
since libOTe is unstable when used from several threads.
We observe that the speed-ups obtained this way are significant
(up to $8.4\times$ for the linear scan and up to $7.1\times$ for clustering), though they are far
from being linear in the number of threads. We attribute it to both of our algorithms being mostly
memory- and network-bound. Overall, the multi-thread mode yields query time under \emph{$6$ seconds} (taking
the single-threaded OT phase into account) for our biggest dataset that consists of \emph{ten million $96$-dimensional vectors.}

\begin{table*}
    \centering
    \begin{tabular}{|c|c|c|c|c|c|c|c|}
    \hline
    \multirow{2}{*}{} & \multirow{2}{*}{\textbf{Algorithm}} & \multirow{2}{*}{\makecell{\textbf{Per-client}\\\textbf{Preprocessing}}} & \multirow{2}{*}{\textbf{OT Phase}} & \multicolumn{4}{c|}{\textbf{Query}} \\
    \cline{5-8}
    &  & & & \textbf{Total} & \textbf{Distances} & \textbf{Top-$k$} & \textbf{ORAM} \\
    \hline
    \multirow{3}{*}{\footnotesize \rotatebox[origin=c]{90}{\textbf{SIFT}}} & Linear scan & None & \makecell{1.83 s / 21.6 s \\ 894 MB} & \makecell{33.3 s / 139 s \\ 4.51 GB} & \makecell{19.8 s / 25.6 s \\ 98.7 MB} & \makecell{13.5 s / 111 s\\ 4.41 GB} & None \\
    \cline{2-8}
    & Clustering & \makecell{12.6 s / 24.7 s \\ 484 MB} & \makecell{0.346 s / 4.34 s \\ 156 MB} & \makecell{8.06 s / 59.7 s \\ 1.77 GB} & \makecell{2.21 s / 3.67 s \\ 56.7 MB} & \makecell{1.96 s / 18.0 s \\ 645 MB} & \makecell{3.85 s / 38.1 s \\ 1.06 GB} \\
    \hline
    \multirow{3}{*}{\footnotesize \rotatebox[origin=c]{90}{\makecell{\textbf{Deep} \\ \textbf{1B-1M}}}} & Linear scan & None & \makecell{1.85 s / 20.6 s \\ 894 MB} & \makecell{28.4 s / 133 s \\ 4.50 GB} &\makecell{14.9 s / 20.6 s \\ 86.1 MB}& \makecell{13.5 s / 112 s\\4.41 GB} & None \\
    \cline{2-8}
     & Clustering & \makecell{11.0 s / 20.6 s \\ 407 MB} & \makecell{0.323 s / 4.09 s \\ 144 MB} & \makecell{6.95 s / 47.8 s \\ 1.58 GB} & \makecell{1.66 s / 3.13 s \\ 44.1 MB} & \makecell{1.93 s / 16.6 s \\ 620 MB} & \makecell{3.37 s / 27.9 s \\ 920 MB}\\
    \hline
    \multirow{3}{*}{\footnotesize \rotatebox[origin=c]{90}{\makecell{\textbf{Deep} \\ \textbf{1B-10M}}}} & Linear scan & None & \makecell{20.0 s / 232 s \\ 9.78 GB} & \makecell{375 s / 1490 s \\ 47.9 GB} & \makecell{202 s / 201 s \\ 518 MB} & \makecell{173 s / 1280 s \\ 47.4 GB} & None\\
    \cline{2-8}
     & Clustering & \makecell{86.0 s / 167 s \\ 3.71 GB} & \makecell{1.04 s / 13.4 s \\ 541 MB} & \makecell{30.1 s / 181 s \\ 5.53 GB} & \makecell{6.27 s / 10.2 s \\ 59.4 MB} & \makecell{7.22 s / 61.0 s \\ 2.35 GB} & \makecell{16.5 s / 107 s \\ 3.12 GB}\\
    \hline
    \multirow{3}{*}{\footnotesize \rotatebox[origin=c]{90}{\textbf{Amazon}}} & Linear scan & None & \makecell{1.99 s / 23.3 s \\ 960 MB}  &\makecell{22.9 s / 133 s \\ 4.85 GB} & \makecell{8.27 s / 14.0 s \\ 70.0 MB} & \makecell{14.6 s / 118 s \\ 4.78 GB} & None\\
    \cline{2-8}
     & Clustering & \makecell{7.27 s / 13.4 s \\ 247 MB} & \makecell{0.273 s / 3.17 s \\ 108 MB} & \makecell{4.54 s / 35.2 s \\ 1.12 GB} & \makecell{0.68 s / 2.31 s \\ 24.4 MB} & \makecell{1.64 s / 13.8 s \\ 528 MB} & \makecell{2.22 s / 18.8 s \\ 617 MB } \\
    \hline
    \end{tabular}
    \caption{Evaluation of \textsf{SANNS} in a single-thread mode. Preprocessing is done once per client, OT phase is done once per query. In each cell, timings are given for fast and slow networks, respectively.}
    \label{single_thread_table}
\end{table*}

\begin{table*}
    \centering
    \begin{tabular}{|c|c|c|c|c|c|c|c|c|c|c|}
        \hline
        & \multirow{2}{*}{\textbf{Algorithm}} & \multicolumn{8}{c|}{\textbf{Threads}} & \multirow{2}{*}{\textbf{Speed-up}}\\
        \cline{3-10}
        & & $1$ & $2$ & $4$ & $8$ & $16$ & $32$ & $64$ & $72$ & \\
        \hline
        \multirow{3}{*}{\footnotesize \rotatebox[origin=c]{90}{\makecell{\textbf{SIFT}}}} & Linear scan & \makecell{33.3 s \\ 139 s} & \makecell{23.2 s \\ 76.4 s} & \makecell{13.4 s \\ 46.9 s} & \makecell{8.04 s \\ 32.5 s} & \makecell{4.78 s \\ 25.7 s} & \makecell{4.25 s \\ 22.1 s} & \makecell{\textbf{3.96 s} \\ \textbf{20.9 s}} & \makecell{4.14 s \\ 21.3 s} & \makecell{8.4 \\ 6.7} \\
        \cline{2-11}
        & Clustering & \makecell{8.06 s \\ 59.7 s} & \makecell{4.84 s \\ 35.2 s} & \makecell{3.16 s \\ 23.6 s} & \makecell{2.18 s \\ 24.4 s} & \makecell{1.65 s \\ 20.1 s} & \makecell{1.55 s \\ 14.2 s} & \makecell{\textbf{1.44 s} \\ \textbf{11.1 s}} & \makecell{1.47 s \\ 12.1 s} & \makecell{5.6 \\ 5.4}\\
        \hline
        \multirow{3}{*}{\footnotesize \rotatebox[origin=c]{90}{\makecell{\textbf{Deep} \\ \textbf{1B-1M}}}} & Linear scan & \makecell{28.4 s \\ 133 s} & \makecell{19.9 s \\ 75.5 s} & \makecell{11.4 s \\ 44.5 s} & \makecell{7.39 s \\ 31.9 s} & \makecell{4.53 s \\ 24.5 s} & \makecell{\textbf{3.94 s} \\ 22.0 s} & \makecell{\textbf{3.94 s} \\ 22.5 s} & \makecell{4.05 s \\ \textbf{21.1 s}} & \makecell{7.2\\6.3}\\
        \cline{2-11}
        & Clustering & \makecell{6.95 s \\ 47.8 s} & \makecell{4.20 s \\ 28.5 s} & \makecell{2.62 s \\ 22.0 s} & \makecell{2.03 s \\ 23.0 s} & \makecell{1.52 s \\ 18.4 s} & \makecell{1.43 s \\ 14.7 s} & \makecell{\textbf{1.37 s} \\ \textbf{11.0 s}} & \makecell{1.39 s \\ 11.5 s} & \makecell{5.1\\4.3}\\
        \hline
        \multirow{3}{*}{\footnotesize \rotatebox[origin=c]{90}{\makecell{\textbf{Deep} \\ \textbf{1B-10M}}}} & Linear scan & \makecell{375 s \\ 1490 s} & \makecell{234 s \\ 800 s} & \makecell{118 s \\ 480 s} & \makecell{81.8 s \\ 343 s} & \makecell{65.8 s \\ 266 s} & \makecell{55.0 s \\ 231 s} & \makecell{\textbf{53.1 s} \\ \textbf{214 s}} & \makecell{58.5 s* \\ 216 s*} & \makecell{7.1\\7.0}\\
        \cline{2-11}
        & Clustering & \makecell{30.1 s \\ 181 s} & \makecell{18.0 s \\ 97.5 s} & \makecell{10.8 s \\ 60.0 s} & \makecell{7.21 s \\ 54.5 s} & \makecell{4.85 s \\ 48.1 s} & \makecell{4.58 s \\ 37.2 s} & \makecell{\textbf{4.23 s} \\ 30.3 s} & \makecell{4.25 s \\ \textbf{28.4 s}} & \makecell{7.1\\6.4}\\
        \hline
        \multirow{3}{*}{\footnotesize \rotatebox[origin=c]{90}{\makecell{\textbf{Amazon}}}} & Linear scan & \makecell{22.9 s \\ 133 s} & \makecell{15.4 s \\ 73.1 s} & \makecell{10.1 s \\ 46.1 s} & \makecell{6.66 s \\ 33.8 s} & \makecell{4.14 s \\ 26.2 s} & \makecell{3.73 s \\ 24.1 s} & \makecell{3.78 s \\ 22.0 s} & \makecell{\textbf{3.64 s} \\ \textbf{21.7 s}} & \makecell{6.3\\6.1}\\
        \cline{2-11}
        & Clustering & \makecell{4.54 s \\ 35.2 s} & \makecell{2.66 s \\ 21.4 s} & \makecell{1.87 s \\ 14.9 s} & \makecell{1.40 s \\ 16.8 s} & \makecell{1.17 s \\ 14.2 s} & \makecell{1.15 s \\ 11.5 s} & \makecell{\textbf{1.12 s} \\ 10.8 s} & \makecell{1.16 s \\ \textbf{9.19 s}} & \makecell{4.1\\3.8}\\
        \hline
    \end{tabular}
    \caption{Evaluation of \textsf{SANNS} query algorithms in the multi-thread mode.
    Each cell contains timings for fast and slow networks. Optimal settings are
    marked in bold. For the numbers marked with an asterisk, we take the \emph{median} of the running times over several runs,
    since with small probability (approximately $20-30\%$), memory swapping starts due to exhaustion of all the available RAM, which affects the running times dramatically (by a factor of $\approx 2\times$).}
    \label{multi_thread_table}
\end{table*}

\begin{table*}
	\centering
	\begin{tabular}{|c|c|c|c|c|}
		\hline
		& \textbf{SIFT} & \textbf{Deep1B-1M} & \textbf{Deep1B-10M} & \textbf{Amazon} \\
		\hline
		\textbf{AHE} & \makecell{19.8 s / 25.6 s \\ 98.7 MB} & \makecell{14.9 s / 20.6 s \\ 56.7 MB}  & \makecell{202 s / 201 s \\ 518 MB}  & \makecell{8.27 s / 14.0 s \\ 70 MB}  \\
		\hline
		{\textbf{OT-based (lower bound)}} & \makecell{19.1 s / 181 s \\ 8.83 GB}  & \makecell{14.5 s / 153 s \\ 6.62 GB}  & \makecell{204 s / 1510 s \\ 66.2 GB}  & \makecell{8.59 s / 88.7 s \\ 3.93 GB}  \\
		\hline
	\end{tabular}
	\caption{Comparison of AHE- and OT-based approach for computing distances. Each cell has two timings: for the fast and the slow networks.}
        \label{ahe_mpc_table}
\end{table*}

\subsection{Microbenchmarks}

As we discussed in the Introduction, all the prior work that has security guarantees similar to \textsf{SANNS}
implements linear scan. 
Thus, in order to provide a detailed comparison,
we compare our approaches in terms of distance computation and top-$k$ against the ones used in the prior work.

\vspace{0.3em}
\noindent {\bf Top-$k$ selection}
We evaluate the new protocol for the approximate top-$k$ selection via garbling the circuit
designed in Section~\ref{sec: approx select}
and compare it with the na\"ive circuit obtained by a direct implementation of Algorithm~\ref{alg: naivetopk}. 
The latter was used in some of the prior work on the secure $k$-NNS~\cite{asharov2018privacy,songhori2015compacting,schoppmann2018private}.
We assume the parties start with arithmetic secret shares of $n = 1\,000\,000$ 24-bit integers. We evaluate both of the above approaches for $k \in \{1, 5, 10, 20, 50, 100\}$.
For the approximate selection, we set the number of bins $l$ such that
on average we return $(1 - \delta) \cdot k$ entries correctly for $\delta \in \{0.01, 0.02, 0.05, 0.1\}$, using the formula from
the proof
of Theorem~\ref{topk_expectation}.
For each setting, we report average running time over slow and fast networks as well as the total communication. 
Table~\ref{table_topk} summarizes our experiments.
As expected, the performance of the approximate circuit is essentially independent of $k$, whereas
the performance of the na\"ive circuit scales linearly as $k$ increases.
Even if we allow the error of only $\delta = 0.01$ (which for $k = 100$ means we return a \emph{single} wrong number),
the performance improves by a factor up to $25$ on the fast network and up to $37$ on the slow network.

The works \cite{SFR18,shaul2018scalable} used fully-homomorphic encryption (FHE) for the top-$k$ selection. 
However, even if we use TFHE~\cite{chillotti2016faster}, which is by far the most efficient FHE approach for highly-nonlinear operations, it will still be several orders of magnitude slower than garbled circuits, since TFHE requires several milliseconds per gate, whereas GC requires less than a microsecond. 
\begin{table*}
\centering
\begin{tabular}{|c|c|c|c|c|c|c|}
\hline
\multirow{2}{*}{\textbf{$k$}} & \multirow{2}{*}{\textbf{Exact}} & \multicolumn{4}{c|}{\textbf{Approximate}} & \multirow{2}{*}{\textbf{Speed-up}} \\
\cline{3-6}
    &       & $\delta = 0.01$ & $\delta = 0.02$ & $\delta = 0.05$ & $\delta = 0.1$ &  \\
\hline
\textbf{1} & \makecell{11.1 s / 93.9 s \\ 3.48 GB} & N/A & N/A & N/A & N/A & N/A \\
\hline
\textbf{5} & \makecell{22.4 s / 249 s \\ 9.62 GB} & \makecell{10.5 s / 90.6 s \\ 3.48 GB} & \makecell{10.6 s / 88.8 s \\ 3.48 GB} & \makecell{10.5 s / 94.5 s \\ 3.48 GB} & \makecell{10.7 s / 90.6 s \\ 3.48 GB} & \makecell{2.1 / 2.7} \\
\hline
\textbf{10} & \makecell{36.1 s / 448 s \\ 17.3 GB} & \makecell{10.7 s / 86.9 s \\ 3.48 GB} & \makecell{10.6 s / 91.2 s \\ 3.48 GB} & \makecell{11.0 s / 89.6 s \\ 3.48 GB} & \makecell{11.0 s / 91.3 s \\ 3.48 GB} & \makecell{3.4 / 5.2} \\
\hline
\textbf{20} & \makecell{67.8 s / 821 s \\ 32.7 GB} & \makecell{10.6 s / 95.2 s \\ 3.50 GB} & \makecell{10.7 s / 94.0 s \\ 3.49 GB} & \makecell{10.8 s / 92.9 s \\ 3.48 GB} & \makecell{10.6 s / 93.8 s \\ 3.48 GB} & \makecell{6.4 / 8.6} \\
\hline
\textbf{50} & \makecell{153 s / 2100 s \\ 78.7 GB} & \makecell{11.1 s / 99.2 s \\ 3.66 GB} & \makecell{10.6 s / 97.4 s \\ 3.57 GB} & \makecell{10.5 s / 94.5 s \\ 3.51 GB} & \makecell{10.5 s / 94.1 s \\ 3.49 GB} & \makecell{14 / 21} \\
\hline
\textbf{100} & \makecell{301 s / 4130 s \\ 156 GB} & \makecell{12.0 s / 113 s \\ 4.22 GB} & \makecell{12.0 s / 98.3 s \\ 3.85 GB} & \makecell{10.8 s / 96.0 s \\ 3.62 GB} & \makecell{11.2 s / 98.6 s \\ 3.55 GB} & \makecell{25 / 37} \\
\hline
\end{tabular}
\caption{Comparison of the exact and the approximate top-$k$ selection protocols (selecting from one million values). Each cell has two timings: for the fast and the slow networks. We report the speed-ups for fast and slow networks between the approximate algorithm with error rate $\delta = 0.01$ and the exact
algorithm.}
\label{table_topk}
\end{table*}

\vspace{0.3em}
\noindent {\bf Distance Computation}
The most efficient way to compute $n$ Euclidean distances securely, besides using the BFV scheme, is arithmetic MPC~\cite{demmler2015aby} based on oblivious transfer (one other alternative used in many prior works~\cite{ppface,sadeghi2009efficient,evans2011efficient,fingercode,elmehdwi2014secure}
is Paillier AHE scheme, which is known to be much less suitable
for the task due to the absence of SIMD capabilities~\cite{gazelle}).
Let us compare BFV scheme used in \textsf{SANNS} with the OT-based distance computation from~\cite{demmler2015aby}
with an optimization from~\cite{mohassel2017secureml}. The latter allows to compute $n$ $l$-bit distances between $d$-dimensional vectors ($l=24$ for
Amazon, $l=23$ for all the other datasets),
using $ndl(l+1)/256$ OTs of $128$-bit strings. We perform those OTs using libOTe for each of our datasets
and measure time (over fast and slow networks) as well as communication.
The results are summarized in Table~\ref{ahe_mpc_table}.
As expected, the communication required by OT-based multiplication is much larger than for AHE (by a factor up to $127\times$). As a result, for the slow network, OT-based multiplication is noticeably slower, by a factor up to $7.5\times$; for the fast network, OT-based approach is no more than $4\%$ faster than AHE.

\subsection{End-to-End Comparison with Prior Work}

We have shown that individual components used by \textsf{SANNS} are extremely competitive compared to the ones proposed by the prior work. 
Here, we provide the end-to-end performance results on the largest dataset we evaluated \textsf{SANNS} on: Deep1B-10M.
For the fast network, our linear scan requires $395$ seconds per query (taking the OT phase into account),
and clustering requires $31$ seconds; for the slow network, it is $1720$ and $194$ seconds, respectively
(see Table~\ref{single_thread_table}).

One issue with a fair comparison with the prior work is that they are done before the recent MPC and HE optimizations became available. Based on the benchmarks in the previous section, one can definitively conclude that the fastest protocol from the prior work is from~\cite{demmler2015aby}.
Namely, we compute distances using OT with the optimization from~\cite{mohassel2017secureml}, and perform the top-$k$ selection using garbled circuit with the na\"ive circuit in Algorithm~\ref{alg: naivetopk}. To estimate the running time of this protocol, we use Table~\ref{ahe_mpc_table} for distances and we run a separate experiment for na\"ive top-$k$ for $n = 10^7$ and $k = 10$.
This gives us \emph{the lower bound} on the running time of $578$ seconds on the fast network and $6040$ seconds on the slow network,
and the lower bound of 240 GB on the communication.

Overall, this indicates that our linear scan obtains a speed-up of $1.46\times$ on the fast network and $3.51\times$ on the slow network. 
The clustering algorithm yields the \textbf{speed-up of $18.5\times$ on the fast network and $31.0\times$ on the slow network}.
The improvement in communication is $4.1\times$ for the linear scan and $39\times$ for the clustering algorithm.

Note that these numbers are based on the lower bounds for the runtime of prior work and several parts of the computation and communication of their end-to-end solution are not included in this comparison.
In particular, just computing distances
using the original implementation from~\cite{demmler2015aby} on SIFT dataset takes $620$ seconds in the fast network, \textit{more than $32\times$ higher compared against our assumed lower bound of $19.1$ seconds in Table~\ref{ahe_mpc_table}}.
When scaling their implementation to ten million points, the system runs out of memory (more than $144$ GB of RAM is needed). 
In conclusion, the speed-up numbers we reported reflect running the best prior algorithm using our new optimized implementation, which leads to a more fair comparison (\textsf{SANNS} speed-up is significantly higher if the \emph{original} implementations of prior works are considered).

\section{Conclusions and Future Directions}
\label{sec: conclusion}

In this work, we design new secure computation protocols for approximate $k$-Nearest Neighbors Search between a client holding a query and a server holding a database, with the Euclidean distance metric. Our solution combines several state-of-the-art cryptographic primitives such as lattice-based AHE, FSS-based distributed ORAM and garbled circuits with various optimizations. Underlying one of
our protocols is a new sublinear-time plaintext approximate $k$-NNS algorithm tailored to secure computation. Notably,
it is the first sublinear-time $k$-NNS protocol
implemented securely. Our performance results show that our solution scales well to massive datasets consisting of up to ten million points. We highlight some directions for future work: 

\begin{itemize}[leftmargin=0mm]
\item Our construction is secure in the semi-honest model, but it would be interesting to extend our protocols to protect against malicious adversaries which can deviate from the protocol.

\item One possible future direction is to implement other sublinear $k$-NNS algorithms securely, most notably Locality-Sensitive Hashing (LSH)~\cite{andoni2015practical}, which has \emph{provable} sublinear query time and is widely used in practice.

\item It is important to study to what extent $k$-NNS queries leak information about the dataset and how much approximation in the answers adds to this leakage.
For instance, the client may try to locate individual points in a dataset by asking several queries that are perturbations of each other and checking if the point of interest ends up in the answer.
For \emph{low-dimensional} datasets there are known strong recovery attacks~\cite{kornaropoulos2019data}, but for the high-dimensional case---which is the focus of this paper---the possibility of such attacks remains open. Besides attacks, an interesting research direction is how to restrict the client (in the number of $k$-NNS queries or the degree of adaptivity) so to minimize the dataset leakage.

That being said, let us state a few simple observations about additional leakage that can happen due to approximation in the results. There are two sources of approximation: approximate top-$k$ selection and clustering-based $k$-NNS algorithm. For the sake of simplicity, let us discuss the effects of these components separately. For the former, one can show that the probability that the element with rank $l > k$ is included in the output is exponentially small in $l - k$. For the latter, we can notice the following. First, we never leak more than the union of the sets of points closest to the query in the clusters whose centers are closest to the query. Second, if the dataset is clusterable (i.e., can be partitioned into clusters with pairwise distances being significantly larger than the diameters of the clusters)
and queries are close to clusters, then the clustering based $k$-NNS algorithm is exact
and there is no additional leakage due to approximation.
\end{itemize}

\section*{Acknowledgments}

We would like to thank the anonymous reviewers for their feedback and helpful comments. This work was partially done while all the authors visited Microsoft Research Redmond.

The second-named author has been supported in part by ERC Advanced Grant ERC-2015-AdG-IMPaCT, by the FWO under an Odysseus project GOH9718N and by the CyberSecurity Research Flanders with reference number VR20192203. Any opinions, findings and conclusions or recommendations expressed in this material are those of the author(s) and do not necessarily reflect the views of the ERC or FWO.

\bibliographystyle{abbrv}
\bibliography{bib}

\appendix 
\section{Chosen Hyperparameters in Clustering-Based Algorithm}
\label{sec:optimal_param}

In Table~\ref{settings_both} and Table~\ref{settings_clustering}, we summarize the parameters we use for both of our algorithms
on each of the datasets.

\begin{table}[h]
    \centering
    \begin{tabular}{|c|c|c|c|c|c|c|c|c|}
        \hline
          & \multicolumn{4}{c|}{Linear scan}&\multicolumn{4}{c|}{Clustering} \\
         \hline
        {\footnotesize \rotatebox[origin=c]{90}{Params}} &
        {\footnotesize \rotatebox[origin=c]{90}{SIFT}} &
		\makecell{{\footnotesize \rotatebox[origin=c]{90}{Deep}}{\footnotesize \rotatebox[origin=c]{90}{1B-1M}}} &
		\makebox{{\footnotesize \rotatebox[origin=c]{90}{Deep}}{\footnotesize \rotatebox[origin=c]{90}{1B-10M}}} &
		{\footnotesize \rotatebox[origin=c]{90}{Amazon}} &
		{\footnotesize \rotatebox[origin=c]{90}{SIFT}} &
		\makecell{{\footnotesize \rotatebox[origin=c]{90}{Deep}} {\footnotesize \rotatebox[origin=c]{90}{1B-1M}}} &
		\makecell{{\footnotesize \rotatebox[origin=c]{90}{Deep}}{\footnotesize \rotatebox[origin=c]{90}{1B-10M}}} &
		{\footnotesize \rotatebox[origin=c]{90}{Amazon}} \\
\hline
         $\lstash$ & 8334 & 8334 & 83 & 8739 & 262 & 210 & 423 & 84 \\
         \hline
         $\bcoord$ & 8 & 8 & 8 & 9 & 8 & 8 & 8 & 9 \\
         \hline
         $\rpoints$ & 8 & 8 & 9 & 7 & 8 & 8 & 8 & 6 \\
         \hline
    \end{tabular}
    \caption{(Near-)optimal hyperparameters that are used both by linear scan and the clustering-based algorithm.}
    \label{settings_both}
\end{table}

\begin{table}[h]
    \centering
    \begin{tabular}{|c|c|c|c|c|}
         \hline
         {\footnotesize \rotatebox[origin=c]{90}{Params}} & 
         {\footnotesize \rotatebox[origin=c]{90}{SIFT}} &
         \makecell{{\footnotesize \rotatebox[origin=c]{90}{Deep}} {\footnotesize \rotatebox[origin=c]{90}{1B-1M}}} & 
         \makecell{{\footnotesize \rotatebox[origin=c]{90}{Deep}} {\footnotesize \rotatebox[origin=c]{90}{1B-10M}}} &
         {\footnotesize \rotatebox[origin=c]{90}{Amazon}} \\
         \hline
         $T$ & 4 & 5 & 6 & 5 \\
         \hline
         $\kcl^i$ & \makecell{50810 \\ 25603 \\ 9968 \\ 4227} & \makecell{44830 \\ 25867 \\ 11795 \\ 5607  2611} & \makecell{209727 \\ 107417 \\ 39132  14424 \\ 5796  2394} &
         \makecell{41293 \\ 24143 \\ 9708 \\ 3516  1156} \\
         \hline
         $m$ & 20 & 22 & 48 & 25 \\
         \hline
         $u^i$ & \makecell{50  31 \\ 19  13} & \makecell{46  31 \\ 19  13  7} & \makecell{88  46  25 \\ 13  7  7} &
         \makecell{37  37 \\ 22  10  7} \\
         \hline
         $s$ & 31412 & 25150 & 50649 & 8228 \\
         \hline
         $l^i$ & \makecell{458 270 \\ 178 84} & \makecell{458 270 \\ 178 84 84} &\makecell{924 458 178 \\ 93 84 84} & \makecell{364 364 \\ 178 84  84}\\
         \hline
         $\rcenters$ & 5 & 5 & 5 & 4 \\
         \hline
         $\alpha$ & 0.56 & 0.56 & 0.56 & 0.56 \\
         \hline
    \end{tabular}
    \caption{(Near-)optimal hyperparameters that are specific to the clustering-based algorithm.}
    \label{settings_clustering}
\end{table}
 \section{Stream Ciphers as PRF}
\label{app:stream}

In the original Floram construction~\cite{doerner2017scaling,Floram_GitLab,ACK}, the PRF and the PRG used in the read-only process are chosen by the authors to be AES-128. 
The implementations of AES are highly optimized, with less than 5000 non-free gates per block~\cite{BP12}. 
As an alternative to AES, the authors also propose  the streams Salsa20~\cite{Salsa20} and its variant  Chacha20~\cite{Chacha20}. 
However, other symmetric ciphers can be used to obtain an efficient PRF/PRG.
In particular, we looked for a PRF with low number of AND gates in order to decrease the communication between the parties when it is evaluated in GC (in the Free-XOR setting).
Some of the most promising constructions are the block cipher LowMC~\cite{LowMC} and the stream cipher Kreyvium~\cite{Kreyvium} (variant of Trivium~\cite{Trivium}). 
In particular Kreyvium is flexible in terms of input and output size, since there is no fixed block size to respect, and its evaluation is very efficient in terms of AND gates per output bit of stream.
The advantage in using Kreyvium starts showing when the size of the inputs starts growing. 
In Table~\ref{tab:AESvsKrey} we estimate the number of AND gates that are executed by the different ciphers for 3 dataset sizes.
We compute 2 PRFs per input, so the actual number of AND gates in Table~\ref{tab:AESvsKrey} should be doubled.

\begin{table}[h]
	\centering
	\resizebox{\columnwidth}{!}{
		\begin{tabular}{|c|c|c|c|}
			\hline
			& $128$ bits & $2.7$ kB & $6$ kB \\
			\hline
			AES-128 & 
			\begin{tabular}{@{}c@{}} 5000 AND \\ (39 AND/bit) \end{tabular} &
			\begin{tabular}{@{}c@{}} 865000 AND \\ (39.1 AND/bit) \end{tabular} &
			\begin{tabular}{@{}c@{}} 1920000 AND \\ (39.06 AND/bit) \end{tabular} \\
			\hline
			Chacha20 & 
			\begin{tabular}{@{}c@{}} 20480 AND \\ (160 AND/bit) \end{tabular} &
			\begin{tabular}{@{}c@{}} 901120 AND \\ (40.7 AND/bit) \end{tabular} &
			\begin{tabular}{@{}c@{}} 1966080 AND \\ (40 AND/bit) \end{tabular} \\
			\hline
			Kreyvium & 
			\begin{tabular}{@{}c@{}} 3840 AND \\ (30 AND/bit) \end{tabular} & 
			\begin{tabular}{@{}c@{}} 69810 AND \\ (3.15 AND/bit) \end{tabular} & 
			\begin{tabular}{@{}c@{}} 150912 AND \\ (3.07 AND/bit) \end{tabular} \\
			\hline
		\end{tabular}
	}
	\vspace{0.1cm}
	\caption{\footnotesize Estimates on the number of AND gates for ciphers AES-128, Chacha20 and Kreyvium for different input sizes.
	The estimates for Chacha20 refer to a naive implementation of the scheme: we believe that the scheme would be more efficient in terms of non trivial gates in practice, but we did not found such optimal estimates in the literature. 
	We do not report the number of AND gates for LowMC: they should be comparable to the estimates we have for Kreyvium for an optimal choice of the parameters.}
	\label{tab:AESvsKrey}
\end{table}

While our approach is more efficient in GC with respect to Floram, the plaintext evaluation of Kreyvium is slower than the (highly optimized) hardware implementation of AES. 
In order to mitigate this issue, we vertically batch 512 bits and we compute multiple streams in parallel (using AVX-512), so we are able to process several hundreds of Mega Bytes of information per second in single core.

\section{Proofs for Approximate Top-$k$}
\label{appendix_topk_proofs}

In this section, we give proofs for Theorem~\ref{topk_expectation}
and Theorem~\ref{topk_whp}.

\begin{proof}[Proof of Theorem~\ref{topk_expectation}]
    First, suppose that we assign a bin for each element uniformly and \emph{independently}.
    For this sampling model, it is not hard to see that the desired expectation of the size of the intersection $\mathcal{I}$ is: \\
    $\mathrm{E} \left[ |\mathcal{I} | \right] = l \cdot \mathrm{Pr}[\mbox{$U_i$ contains at least one of the top-$k$ elements}] \\
     = l \cdot \left(1 - \left(1 - \frac{1}{l}\right)^k\right)
    $,
    where the first step follows from the linearity of expectation,
    and the second step is an immediate calculation.
    Suppose that $l = k / \delta$, where
    $\delta > 0$ is sufficiently small, and suppose that $k \to \infty$. \\ 
    Then, continuing the calculation, \\
   	$
    l \cdot \left(1 - \left(1 - \frac{1}{l}\right)^k\right) = \frac{k}{\delta} \cdot \left(1 - e^{k \cdot \ln \left(1 - \frac{\delta}{k}\right)}\right) \\
    = \frac{k}{\delta} \left(1 - e^{-\delta + O(1 / k)}\right) = \frac{k \cdot (1 - e^{-\delta})}{\delta} + O(1) \\
    \geq \frac{k \cdot \left(\delta - \frac{\delta^2}{2}\right)}{\delta} + O(1) = k \cdot \left(1 - \frac{\delta}{2}\right) + O(1),
    $ \\
    where the second step uses the Taylor series of $\ln x$,
    the third step uses the Taylor series of $e^x$ and
    the fourth step uses the inequality $e^{-x} \leq 1 - x + \frac{x^2}{2}$,
    which holds for small $x>0$ .
    
    To argue about the actual sampling process, where instead of uniform
    and independent
    assignment we shuffle elements and partition them into
    $l$ blocks of size $n / l$, we use
    the main result of~\cite{diaconis1980finite}.
    Namely, it is true that the probability
    $$
    \mathrm{Pr}[\mbox{$U_i$ contains at least one of the top-$k$ elements}]
    $$
    can change by at most $O(1 / l)$ when passing between two sampling processes.
    This means that the overall expectation changes by at most $O(1)$,
    and is thus still at least:
    $
    k \cdot \left(1 - \frac{\delta}{2}\right) + O(1).
    $
    For a fixed $\delta$, this expression is at least $(1 - \delta) k$,
    whenever $k$ is sufficiently large.
\end{proof}

\begin{proof}[Proof of Theorem~\ref{topk_whp}]
As in the proof of the previous theorem,
we start with a simpler sampling model, where bins are
assigned independently. Suppose that $\delta > 0$
is fixed and $k$ tends to infinity. We set $l = k^2 / \delta$.
In that case, one has: \\
$
\mathrm{Pr}[\mbox{all top-$k$ elements end up into different bins}] \\
= \left(1 - \frac{1}{l}\right) \cdot \left(1 - \frac{2}{l}\right) \cdot
\ldots \cdot \left(1 - \frac{k - 1}{l}\right) \\
= \left(1 - \frac{\delta}{k^2}\right) \cdot \left(1 - \frac{2 \delta}{k^2}\right) \cdot \ldots \cdot \left(1 - \frac{(k - 1) \cdot \delta}{k^2}\right) \\
= \exp\left(\ln\left(1 - \frac{\delta}{k^2}\right) + \ln\left(1 - \frac{2 \delta}{k^2}\right) + \ldots \right. 
\left. + \ln\left(1 - \frac{(k - 1) \cdot \delta}{k^2}\right)\right) \\
= \exp\left(- \frac{\delta (1 + 2 + \ldots + (k - 1))}{k^2} + O\left(\frac{1}{k}\right)\right) \\
= e^{-\delta / 2} + O\left(\frac{1}{k}\right)
\geq 1 - \frac{\delta}{2} + O\left(\frac{1}{k}\right),
$ \\
where the fourth step uses the Taylor series of $\ln x$
and the sixth step uses the inequality $e^{-x} \geq 1 - x$.
The final bound is at least $1 - \delta$ provided that
$k$ is large enough.

Now let us prove that for the actual sampling procedure (shuffling
and partitioning into $l$ blocks of size $n / l$), the probability
of top-$k$ elements being assigned to different bins \emph{can only increase},
which implies the desired result. To see this, let us denote $c_i$ the bin
of the $i$-th of the top-$k$ elements. One clearly has:
$
\mathrm{Pr}[\mbox{all top-$k$ elements end up into
    different bins}] = \\
\sum_{\mbox{distinct $j_1, j_2, \ldots, j_k$}} \mathrm{Pr}[
    c_1 = j_1 \wedge c_2 = j_2 \wedge \ldots \wedge c_k = j_k].
$
Thus, it is enough to show that any probability of the form
$
\mathrm{Pr}[
    c_1 = j_1 \wedge c_2 = j_2 \wedge \ldots \wedge c_k = j_k],
$
where $j_1, j_2, \ldots, j_k$ are distinct, can only increase when passing
to the actual sampling model. This probability can be factorized as follows:
$
\mathrm{Pr}[
    c_1 = j_1 \wedge c_2 = j_2 \wedge \ldots \wedge c_k = j_k] \\
= \mathrm{Pr}[
    c_1 = j_1] \cdot \mathrm{Pr}[
    c_2 = j_2 \mid c_1 = j_1] \cdot \ldots \\
\cdot
    \mathrm{Pr}[
    c_k = j_k \mid c_1 = j_1 \wedge \ldots \wedge c_{k-1} = j_{k-1}].
$ \\
For the simplified sampling model,
each of these conditional probabilities is equal to $1 / l$ due
to the independence of $c_i$.
However, for the actual model, they are larger: if we condition
on $t$ equalities, then the probability is equal to $\frac{n}{l (n - t)}$.
This implies the required monotonicity result.
\end{proof}

\section{Security Proofs}
\label{sec: proof}

We prove simulation-based security for our protocols for approximate $k$-NNS. First,
we recall the definition (see e.g. \cite{lindell2017simulate}) of two party computation and simulation-based security for semi-honest adversaries.
\begin{definition} 
A two-party functionality is a possibly randomized function
$f : \{0, 1\}^* \times \{0, 1\}^* \to  \{0, 1\}^* \times \{0, 1\}^*$, that is, for every pair of inputs $x, y \in \{0, 1\}^n$
, the output-pair is a random variable $(f_1(x, y), f_2(x, y))$. The first party (with input $x$) obtains $f_1(x, y)$ and the second party (with input $y$) obtains $f_2(x, y)$.
\end{definition}

Let $\pi$ be a protocol computing the functionality $f$. The {\it view} of the $i$-th party during an execution of $\pi$ on $(x, y)$ and security parameter $\lambda$ is denoted by $\mathrm{View}_{\pi, i}(x,y, \lambda)$
and equals the party $i$'s input, its internal randomness, plus all messages it receives  during the protocol.

\begin{definition} 
Let $f = (f_1, f_2)$ be a functionality and let $\pi$ be a protocol that computes $f$. We say that $\pi$ securely computes $f$ in the presence of static semi-honest adversaries if there exist probabilistic polynomial-time algorithms $S_1$
and $S_2$ (often called {\it simulators}) such that \\
$(S_1(1^{\lambda}, x, f_1(x,y)), f(x,y))
\approx
(\mathrm{View}_{\pi, 1}
(x, y, \lambda), f(x,y))$	
and 
$(S_2(1^{\lambda}, y, f_2(x,y)), f(x,y))
\approx
(\mathrm{View}_{\pi, 2}
(x, y, \lambda), f(x,y))$.
Here $\approx$ means computational indistinguishability. 
\end{definition}

\subsection{Ideal Functionalities}

First, we define the ideal functionalities that our protocol achieves. Note that the two protocols have slightly different ideal functionalities. We will denote them by $\Fcl$ (for clustering) and $\Fls$ (for linear scan).

\begin{figure}[ht!] \footnotesize
	\framebox{\begin{minipage}{0.95\linewidth}
			Parameters: number of elements $n$, dimension $d$, bits of precision $b_c$.
			\begin{itemize}
				\item Input: client inputs a query $\pt{q}{} \in \mathbb{R}^d$ and server inputs database $DB = [(\pt{p}{i}, \ID_i)]_{i=1}^n$. Note that points are truncated to $b_c$ bits. 
			    \item Output: returns output of Algorithm~\ref{alg: linearscan} to client. 
				\end{itemize}
	\end{minipage}}
	\caption{Ideal functionality $\Fls$}
\end{figure}

\begin{figure}[ht!] \footnotesize
	\framebox{\begin{minipage}{0.95\linewidth}
			Parameters: number of elements $n$, dimension $d$,  bits of precision $b_c$, and clustering-based hyperparameters $T$, $\kcl^i$, $m$, $u^i$, $s$, $l^i$, $\lstash$, $\bcoord$, $\rcenters$ and $\rpoints$. 
			\begin{itemize}
				\item Input: client inputs a query $\pt{q}{} \in \mathbb{R}^d$ and server inputs database $DB = [(\pt{p}{i}, \ID_i)]_{i=1}^n$. The points are truncated to $b_c$ bits. 
			    \item Output: returns output of Algorithm~\ref{alg: clustering} to client. 
				\end{itemize}
	\end{minipage}}
	\caption{Ideal functionality $\Fcl$}
\end{figure}

\subsection{Proofs}

\begin{theorem}
Assuming the hardness of the decision-RLWE problem, our linear scan protocol $\Pls$ securely implements the functionality $\Fls$ in the $\cF_{\mathrm{aTOPk}}$ hybrid model, with semi-honest adversaries. 
\end{theorem}

\begin{proof}
First, we construct a simulator for the client. The simulator generates a key $sk$ for the AHE scheme and sends $sk$ to the client. Then, it simulates the server's first message as $\mathsf{AHE.Enc}(sk, r_i)$ for random values $r_i$. From the circuit privacy property of the AHE scheme, this is indistinguishable from the first message in the real protocol. 
Next, the simulator simply feeds $\{r_i\}$ to the ideal functionality $\cF_{\mathrm{aTOPk}}$ and forwards the output to the client. This completes the simulation. 

Next, we construct a simulator for the server. The simulator generates a key $sk$ for the AHE scheme. The first message from the client to the server consists of the encryptions $\mathsf{AHE.Enc}(sk, \pt{q}{}[i])$ in the real protocol. Instead, the simulator just sends $\mathsf{AHE.Enc}(sk, 0) $ for $1 \leq i \leq d$. Based on the RLWE assumption, these views are indistinguishable. Finally, the simulator generates random vector $R = (r_1, \ldots, r_n)$ and sends that to the server. 
\end{proof}

\begin{theorem}
Assuming the hardness of the decision-RLWE problem, our clustering protocol $\Pcl$ securely implements the $\Fcl$ functionality in the ($\cF_{\mathrm{TOPk}}$, $\cF_{\mathrm{aTOPk}}$, $\cF_{\mathrm{DROM}}$)-hybrid model in the presence of semi-honest adversaries. 
\end{theorem}

\begin{proof}
Again correctness is easy to verify. We first describe simulator for the client. First, the simulator generates a secret key $sk$ for the AHE scheme and sends $sk$ to the client. Next, the simulator sends blocks of zero to $\cF_{\mathrm{DROM}}.\mathsf{Init}$. Then, on receiving the query message from the client, the simulator does the following: for each $i, j$, it samples random values $r_{ij}$ and generates $\mathsf{AHE.Enc}(sk, r_{ij})$. Using a similar argument as in the previous proof, these ciphertexts are indistinguishable from the client's view of the first step of the secure protocol. Then, the simulator forwards the $r_{ij}$ to $\cF_{\mathrm{aTOPk}}$ and gets back secret shares of indices, namely $[i_1], \ldots, [i_u]$. Then, it feeds these indices shares to $\cF_{\mathrm{DROM}}.\mathsf{Read}$ and forwards the output to the client. Also, it samples random messages $s_i$ and sends $\mathsf{AHE.Enc}(sk, s_{i})$ to the client. Later, when the simulator receives the shares $m \cdot \uall+s$ of (point, ID) pairs from the client, it samples $m \cdot \uall+s$ random pairs of values and send the first $m \cdot \uall$ values to $\cF_{\mathrm{TOPk}}$ and the last $s$ values to $\cF_{\mathrm{aTOPk}}$. Then, it forwards the output to the client. Since the intermediate values revealed to the client are all independent uniformly random values, the view generated from simulator is indistinguishable from the real view. 
Now, the simulator for server works in almost the same fashion, with the difference that in contrast to the real client which sends $\mathsf{AHE.Enc}(sk, \pt{q}{i})$ for $1 \leq i \leq d$, the simulator simply sends $d$ encryption of zeros. This is indistinguishable from uniform, based on the RLWE assumption. 
\end{proof}

\end{document}